\documentclass[12pt,fleqn,review,authoryear]{csda}
\usepackage{amsmath}
\usepackage{graphicx,psfrag,epsf}
\usepackage{enumerate}
\usepackage{natbib}
\usepackage{url}

\RequirePackage[letterpaper, portrait, margin=1in]{geometry}

\usepackage[mathscr]{eucal}
\usepackage{multirow}
\usepackage{amsfonts,bm}
\usepackage[toc,page]{appendix}
\usepackage[colorlinks,citecolor=blue,urlcolor=blue]{hyperref}
\usepackage{subcaption}
\usepackage{enumitem}
\usepackage{indentfirst}

\usepackage{bbm,mathtools,amsfonts,amsmath,amssymb,graphicx,bm,latexsym,natbib}
\usepackage[usenames,dvipsnames]{color}

\newcommand{\mple}{\widetilde\btheta}

\usepackage{amsthm}

\newcommand{\bthetas}{\btheta^\star}
\newcommand{\bthetah}{\widehat\btheta}

\newcommand{\qqq}{d} 
\newcommand{\rrr}{q} 
\newcommand{\norm}[1]{\lVert#1\rVert}

\usepackage[mathscr]{eucal}

\usepackage{centernot}
\usepackage{dlfltxbcodetips,bbm,mathtools,amsfonts,amsmath,amssymb,fancybox,graphicx,bm,latexsym}

\newcommand{\zs}{\bz^\star}

\newcommand{\zh}{\widehat{\bz}}

\newcommand{\lip}{{\tiny\mbox{Lip}}}

\newcommand{\logit}{\mbox{logit}}

\newcommand{\btt}{\begin{box}}
\newcommand{\ett}{\end{box}}

\newcommand{\define}{\mathop{=}\limits^{\mbox{\tiny def}}}

\usepackage[colorlinks,citecolor=blue,urlcolor=blue]{hyperref}

\DeclareMathOperator*{\argmax}{arg\,max}

\newcommand{\benum}{\begin{enumerate}}
\newcommand{\eenum}{\end{enumerate}}

\newcommand{\bq}{\begin{quote}\em}
\newcommand{\eq}{\end{quote}}
\newcommand{\bbq}{\begin{quote}\bf\em}
\newcommand{\ebq}{\end{quote}}

\newcommand{\ind}{\msim\limits^{\mbox{\tiny ind}}}
\newcommand{\iid}{\msim\limits^{\mbox{\tiny iid}}}

\renewcommand{\=}{&=&}

\newcommand{\lte}{&\leq&}
\newcommand{\gte}{&\geq&}
\newcommand{\mR}{\mathbb{R}}

\newcommand{\mbR}{\mathbb{R}}
\newcommand{\mbN}{\mathbb{N}}
\newcommand{\mbX}{\mathbb{X}}

\newcommand{\one}{\mathbbm{1}}

\newcommand{\mbE}{\mathbb{E}}

\newcommand{\mbZ}{\mathbb{Z}}

\newcommand{\mbS}{\mathbb{S}}
\newcommand{\mbP}{\mathbb{P}}

\newcommand{\hide}[1]{}
\newcommand{\ghost}[1]{}
\newcommand{\ba}{\begin{array}{llllllllll}}
\newcommand{\ea}{\end{array}}
\newcommand{\bea}{\begin{equation}\begin{array}{llllllllll}}
\newcommand{\eea}{\end{array}\end{equation}}
\newcommand{\be}{\begin{equation}\begin{array}{lllllllllllllllll}}
\newcommand{\beno}{\begin{equation}\begin{array}{lllllllllllll}\nonumber}
\newcommand{\ee}{\end{array}\end{equation}}
\newcommand{\bel}{\begin{equation}\begin{array}{lllllllllllll}\nonumber}
\newcommand{\eel}{\Box\end{array}\end{equation}}
\newcommand{\bi}{\begin{itemize}}
\newcommand{\ei}{\end{itemize}}
\newcommand{\ben}{\begin{enumerate}}
\newcommand{\een}{\end{enumerate}}
\newcommand{\alert}{\textcolor{black}}
\newcommand{\dis}{\displaystyle}
\newcommand{\dsum}{\displaystyle\sum\limits}

\newcommand{\dprod}{\displaystyle\prod\limits}

\newcommand{\mS}{\mathbb{S}}

\newcommand{\mA}{\mathscr{A}}

\newcommand{\s}{\vspace{0.25cm}}
\newcommand{\bx}{\bm{x}}

\newcommand{\bX}{\bm{X}}
\newcommand{\mX}{\mathbb{X}}

\newcommand{\mZ}{\mathbb{Z}}

\newcommand{\mL}{\mathbb{L}}

\newcommand{\bz}{\bm{z}}
\newcommand{\bZ}{\bm{Z}}

\newcommand{\balpha}{\bm{\alpha}}

\newcommand{\bTheta}{\bm\Theta}
\newcommand{\bta}{\bm{\eta}}
\newcommand{\btheta}{\boldsymbol{\theta}}

\newcommand{\bpi}{\mbox{\boldmath$\pi$}}

\newcommand{\msim}{\mathop{\rm \sim}}

\newcounter{comment}
\newenvironment{comment}[1][]{\refstepcounter{comment}\par\smallskip\noindent%
\textbf{Comment~\thecomment #1}:\vspace{0.25cm}\\ \rmfamily}{}
\newcommand{\bc}{\begin{comment}\em}
\newcommand{\ec}{\end{comment}}

\newcounter{ex}

\newcounter{counterexample}

\newcounter{definition}

\newcounter{theorem}
\newenvironment{theorem}[1][]{\refstepcounter{theorem}\par\smallskip\indent%
\textbf{Theorem~\thetheorem #1}.\em \rmfamily}{}

\newcounter{ttproof}
\newcommand{\ttproof}{%
\addtocounter{ttproof}{1}%
{\medskip}\textsc{Proof of Theorem}{\medskip}
}

\newcounter{pproof}
\newcommand{\pproof}{%
\addtocounter{pproof}{1}%
{\medskip}\textsc{Proof of Proposition}{\medskip}
}

\newcounter{corollary}

\newcounter{ccproof}
\newcounter{llproof}
\newcommand{\llproof}{%
\addtocounter{llproof}{1}%
{\medskip}\textsc{Proof of Lemma}{\medskip}
}

\newcounter{lemma}
\newenvironment{lemma}[1][]{\refstepcounter{lemma}\par\medskip\indent%
\textbf{Lemma~\thelemma #1}.\em \rmfamily}{}

\newcounter{proposition}
\newenvironment{proposition}[1][]{\refstepcounter{proposition}\par\medskip\indent%
\textbf{Proposition~\theproposition #1}.\em \rmfamily}{}

\newcounter{example}

\newcounter{com}

\newcounter{defin}

\newcounter{assumption}

\newcommand{\mQ}{M}
\renewcommand{\mL}{\mathscr{L}}

\newcommand{\longtitle}{Large-scale estimation of random graph models with local dependence}

\begin{document}

\begin{frontmatter}

\title{\longtitle}

\author{Sergii Babkin}
\address{Microsoft}

\author{Jonathan Stewart}
\address{Department of Statistics,
Rice University}

\author{Xiaochen Long}
\address{Department of Statistics,
Rice University}

\author{Michael Schweinberger\corref{cor1}}
\address{Department of Statistics,
Rice University}

\journal{Computational Statistics \& Data Analysis}

\begin{abstract}
A class of random graph models is considered,
combining features of exponential-family models and latent structure models,
with the goal of retaining the strengths of both of them while reducing the weaknesses of each of them.
An open problem is how to estimate such models from large networks.
A novel approach to large-scale estimation is proposed,
taking advantage of the local structure of such models for the purpose of local computing.
The main idea is that random graphs with local dependence can be decomposed into subgraphs,
which enables parallel computing on subgraphs and suggests a two-step estimation approach.
The first step estimates the local structure underlying random graphs.
The second step estimates parameters given the estimated local structure of random graphs.
Both steps can be implemented in parallel,
which enables large-scale estimation.
The advantages of the two-step estimation approach are demonstrated by simulation studies with up to 10,000 nodes and an application to a large Amazon product recommendation network with more than 10,000 products.
\end{abstract}


\end{frontmatter}

\section{Introduction} 
\label{sec:introduction}

The statistical analysis of network data is an emerging area in statistics \citep{Ko09a}. 
Network data arise in the study of insurgent and terrorist networks,
contact networks facilitating the spread of infectious diseases (e.g., corona viruses such as SARS and MERS; Ebola; HIV),
social networks,
and the World Wide Web.

Many models of network data have been proposed,
as described in recent review papers \citep{Fien12,HuKrSc12,latentspace,SmAsCa19,ScKrBu17,Ho18}.
Among the plethora of models,
two broad classes of models can be distinguished:
models with latent structure,
including stochastic block models \citep[e.g.,][]{NkSt01,BiCh09,RoChYu11} and latent space models \citep[e.g.,][]{HpRaHm01,HaRaTa07,sewell2015latent,SmAsCa19};
and exponential-family models of random graphs \citep{ergm.book,Ha13},
including Erd\H{o}s and R\'enyi random graphs,
logistic regression models,
and models resembling Markov random fields in spatial statistics \citep{Bj74}.
Both have advantages and disadvantages,
as one of the pioneers of statistical network analysis observed:
\begin{quote}\em
``I expect that, especially for modelling larger networks (with, say, a few hundred or more nodes), the latent space models will not be able to represent network structures as expressed by subgraph counts sufficiently well and the exponential random-graph models will not be able to represent the cohesive structure sufficiently well. Models that combine important features of these two approaches may be the next generation of social network models" \citep[][p.\ 324]{Sn07a}.
\end{quote}
In other words,
latent structure models capture who is close to whom,
but are not flexible models of many network phenomena of interest \citep[despite the fact that latent space models induce a weak tendency towards transitivity, see][]{HpRaHm01}.
Well-posed exponential-family models of random graphs can capture a vast range of network phenomena of interest (including, but not limited to, transitivity),
but the underlying assumptions of many of those models make more sense in small networks than large networks.
A case in point is the Markov random graphs of \citet{FoSd86}.
These models assume that possible edges $X_{i,j} \in \{0, 1\}$ and $X_{k,l} \in \{0, 1\}$ of pairs of nodes $\{i, j\}$ and $\{k, l\}$ are independent conditional on the rest of the graph provided $\{i, j\}$ and $\{k, l\}$ do not overlap,
but are otherwise dependent.
As a consequence,
each possible edge $X_{i,j}$ depends on $2\, (n-2)$ other possible edges,
where $n$ is the number of nodes.
If Markov random graphs were applied to online social networks such as Facebook,
then each possible friendship would depend on billions of other possible friendships.
Such dependence assumptions are implausible,
and when coupled with homogeneity assumptions regarding parameters can give rise to the well-studied issue of model near-degeneracy \citep{Ha03,Sc09b,ChDi11,Me17}.

One class of next-generation random graph models was introduced in \citet{ScHa13},
which combines important features of stochastic block models and exponential-family models,
with the goal of retaining the strengths of both of them while reducing the weaknesses of each of them.
The basic idea is that a set of nodes is partitioned into blocks,
and edges among nodes within and between blocks are governed by exponential-family models of random graphs with local dependence within blocks.
A simple example is a model where edges between blocks are independent Bernoulli random variables,
whereas edges within blocks are generated by an exponential-family model which encourages triangles within blocks but ensures that,
for each pair of nodes within a block,
the added value of additional edges and triangles decays.
Such models induce local dependence within blocks and the overall dependence induced by the model is weak provided the blocks are not too large.
We have shown elsewhere that such models are well-behaved---in contrast to the infamous triangle model first studied by \citet{St86},
\citet{Jo99},
\citet{HaJo99},
and others \citep[e.g.,][]{ChDi11}---and that statistical inference is possible and supported by statistical theory:
e.g.,
we have established concentration and consistency results for canonical and curved exponential-family models of random graphs with local dependence under the assumption that the blocks are known and the sizes of the blocks are similar in a well-defined sense \citep{ScSt16}.
In addition,
when the blocks are unknown,
the block memberships of most nodes can be recovered with high probability under weak dependence and smoothness conditions \citep{Sc17}.

While some progress has been made in terms of statistical theory,
computing remains challenging.
When the block structure is known \citep[which is the case in multilevel networks, e.g., networks consisting of units of armed forces,][]{multilevelnetwork},
statistical inference for parameters can rely on existing methods for estimating parameters of exponential-family models of random graphs \citep[e.g.,][]{StIk90,Sn02,HuHa04,CaFr09,HuHuHa12,OkGe12,AtLaRo12,Jin2013p927,ThKa17,Kr17,Byetal18,TaFr18}.
If the block structure is unknown,
however,
it needs to be estimated based on the observed network.
The recovery of unknown block structure resembles the recovery of unknown block structure in stochastic block models \citep[e.g.,][]{NkSt01,BiCh09,BiChLe11,RoChYu11,ChWoAi12,CeDaLa11,PrSuTaVo12,AmChBiLe13,RoQiFa13,ZhZh16,BiVoRo17,Gaetal17}.
However,
the recovery of unknown block structure is more challenging for exponential-family models with local dependence than for stochastic block models.
The main challenge is the intractability of the complete-data likelihood function,
i.e., 
the likelihood function given an observation of the network as well as the block structure.
The intractability of the complete-data likelihood function is rooted in the intractability of the normalizing constants of within-block probability mass functions,
which stems from the local dependence within blocks.

We present here a tractable approximation of the likelihood function,
leveraging concentration results for random graphs with local dependence.
Based on the approximation of the likelihood function,
we propose a two-step estimation approach that exploits the local structure of random graphs with local dependence for the purpose of local computing.
The first step estimates the block structure and decomposes the random graph into subgraphs with local dependence.
The decomposition of the random graph relies on approximations of the likelihood function supported by theoretical results.
The second step estimates parameters given the estimated block structure by using Monte Carlo maximum likelihood methods \citep{HuHa04} or maximum pseudolikelihood methods \citep{StIk90}.
Both steps can be implemented in parallel,
which enables large-scale estimation on multi-core computers or computing clusters.
We demonstrate the advantages of the two-step estimation approach by simulation studies with up to 10,000 nodes and an application to a large Amazon product recommendation network with more than 10,000 products.

The remainder of our paper is organized as follows. 
Section~\ref{sec:hergm} introduces models.
Section \ref{sec:approximation} discusses likelihood-based inference based on approximations of the likelihood function motivated by theoretical results, 
and Section~\ref{sec:variational} takes advantage of such approximations to estimate models.
Section~\ref{sec:sim_full} presents simulation results and Section~\ref{sec:application} presents an application.

\section{Models}
\label{sec:hergm} 

We consider random graphs with a set of nodes $\mA = \{1, \dots, n\}$ and a set of edges $\mathscr{E} \subset \mA \times \mA$.
Here, 
edges are regarded as random variables $X_{i,j} \in \{0, 1\}$,
where $X_{i,j} = 1$ if nodes $i$ and $j$ are connected by an edge and $X_{i,j} = 0$ otherwise.
We focus on undirected random graphs without self-edges,
although the methods introduced here can be extended to directed random graphs.
\alert{Throughout,
we denote by $\bX = (X_{i,j})_{i<j}^n$ the set of possible edges of a random graph, 
and by $\mbX = \{0, 1\}^{\binom{n}{2}}$ the sample space of $\bX$.}

\alert{Since the pioneering research of \citet{HpLs70,HpLs72,HpLs76} in the 1970s,
it is known that network data are dependent data.
In small networks, 
it is very well possible that each possible edge depends on many other possible edges,
and might in fact depend on all other possible edges.
In large networks,
however,
it is not credible that each possible edge depends on all other possible edges:
e.g.,
in networks with $n \gg $ 10,000 nodes---including the network used in Section \ref{sec:application}---it is implausible that each possible edge depends on all $\binom{n}{2} - 1 \gg$ 50,000,000 other possible edges.
Instead,
it is tempting to believe that the dependence among edges in large networks is local,
in the sense that each edge depends on a subset of other possible edges.}
We consider here a simple form of local dependence,
following \citet{ScHa13}.
Assume that $\mA$ is partitioned into $K \geq 2$ subsets of nodes $\mA_1, \dots, \mA_K$,
called blocks,
and let $\bz = (\bz_1, \dots, \bz_n)$ be the block memberships of nodes,
where $z_{i,k} = 1$ if node $i$ belongs to block $\mA_k$ and $z_{i,k} = 0$ otherwise.
We henceforth denote by $\bX_{k,k} = (X_{i,j})_{i < j:\; z_{i,k} = z_{j,k} = 1}^n$ the within-block edges among nodes in block $\mA_k$ ($k = 1, \dots, K$) and by $\bX_{k,l} = (X_{i,j})_{i, j:\; z_{i,k} = z_{j,l} = 1}^n$ the between-block edges among nodes  in block $\mA_k$ and nodes in block $\mA_l$ ($l < k = 1, \dots, K$). 

\s

\alert{\textbf{Definition. Local dependence.}
A model of a random graph $\bX$ satisfies local dependence if the probability mass function of $\bX$ can be factorized as follows:
\be
\label{locdef}
p_{\bta(\btheta, \bz)} (\bx)
&=& \dprod_{k=1}^K p_{\bta_{W,k}(\btheta, \bz)}(\bx_{k,k})\; \dprod_{l=1}^{k-1}\;\, \dprod_{i, j:\; z_{i,k} = z_{j,l} = 1}^n\, p_{\bta_{B,k,l}(\btheta, \bz)}(x_{i,j}),
&& \bx \in \mbX,
\ee
where 
\bi
\item $p_{\bta(\btheta, \bz)} (\bx)$ is the probability mass of graph $\bx$,
parameterized by $\bta(\btheta, \bz) \in \mR^d$;
\item $p_{\bta_{W,k}(\btheta, \bz)}(\bx_{k,k})$ is the probability mass of within-block subgraph $\bx_{k,k}$,
parameterized by $\bta_{W,k}(\btheta, \bz) \in \mR^{d_{W,k}}$ ($k = 1, \dots, K$);
\item $p_{\bta_{B,k,l}(\btheta, \bz)}(x_{i,j})$ is the probability mass of between-block edge $x_{i,j}$,
parameterized by $\bta_{B,k,l}(\btheta, \bz) \in \mR^{d_{B,k,l}}$ ($i, j:\; z_{i,k} = z_{j,l} = 1$);
\item the parameter vector $\bta(\btheta, \bz) \in \mR^d$ consists of the subvectors $\bta_{W,k}(\btheta, \bz) \in \mR^{d_{W,k}}$ ($k = 1, \dots, K$) and $\bta_{B,k,l}(\btheta, \bz) \in \mR^{d_{B,k,l}}$ ($l < k = 1, \dots, K$).
\ei
The map $\bta: \bTheta \times \mZ \mapsto \mR^{\qqq}$ depends on the model.
In general,
$\bta: \bTheta \times \mZ \mapsto \mR^{\qqq}$ may be a linear or non-linear function of a parameter vector $\btheta \in \bTheta \subseteq \mR^{\rrr}$ of dimension $\rrr \leq \qqq$.
An example is given by the curved exponential-family parameterizations used in Section \ref{sec:application},
where $\bta: \bTheta \times \mZ \mapsto \mR^{\qqq}$ is a non-linear function of a parameter vector $\btheta \in \bTheta \subseteq \mR^{\rrr}$ of dimension $\rrr < \qqq$.
Well-chosen curved exponential-family parameterizations help ensure that the added value of additional edges, triangles, and other subgraph configurations decays \citep{StScBoMo18},
which is plausible in many applications and can improve the in-sample performance of models \citep{HuGoHa08} as well as the out-of-sample performance of models \citep{StScBoMo18}.
The factorization of the probability mass function of models with local dependence has at least three implications.
First,
edges among nodes in block $\mA_k$ can depend on other edges in block $\mA_k$ ($k = 1, \dots, K$).
Second,
the factorization implies that edges among nodes in block $\mA_k$ do not depend on edges that involve nodes in other blocks ($k = 1, \dots, K$).
Third,
edges between blocks are independent.
In other words,
the dependence is local in the sense that it is confined to within-block subgraphs.}

We introduce here two examples of models with local dependence,
to be used as motivating examples throughout Sections \ref{sec:approximation}---\ref{sec:sim_full}.

\s

\alert{{\bf Example 1. Stochastic block model.}
An important special case of models with local dependence is given by stochastic block models \citep[e.g.,][]{NkSt01,BiCh09,RoChYu11}.
Stochastic block models assume that possible edges $X_{i,j}$ between nodes $i$ and $j$ in blocks $k$ and $l$ are independent Bernoulli$(\mu_{k,l})$ random variables,
where the probability $\mu_{k,l} \in (0, 1)$ of an edge depends on the blocks $k$ and $l$ of nodes $i$ and $j$,
respectively.
Stochastic block models are special cases of models with local dependence,
having within-block probability mass functions of the form
\beno
p_{\bta_{W,k}(\btheta, \bz)}(\bx_{k,k})
&=& \dprod_{i<j:\; z_{i,k}=z_{j,k}=1}^n \mu_{k,k}^{x_{i,j}}\, (1 - \mu_{k,k})^{1-x_{i,j}}\s\s
\\
&=& \dprod_{i<j:\; z_{i,k}=z_{j,k}=1}^n \exp\left(\theta_{W,k,k}\, x_{i,j} - \log(1 + \exp(\theta_{W,k,k}))\right)\s\s
\\
&\propto& \exp\left(\theta_{W,k,k}\, \dsum_{i<j}^n x_{i,j}\, z_{i,k}\, z_{j,k}\right),
\ee
where
\beno
\theta_{W,k,k}
&=& \logit(\mu_{k,k}) 
&\in& \mR,
&& k = 1, \dots, K.
\ee
The between-block probability mass functions are of the form
\beno
p_{\bta_{B,k,l}(\btheta, \bz)}(\bx_{k,l})
&=& \dprod_{i<j:\; z_{i,k}=z_{j,l}=1}^n \mu_{k,l}^{x_{i,j}}\, (1 - \mu_{k,l})^{1-x_{i,j}}\s\s
\\
&=& \dprod_{i<j:\; z_{i,k}=z_{j,l}=1}^n \exp\left(\theta_{B,k,l}\, x_{i,j} - \log(1 + \exp(\theta_{B,k,l}))\right)\s\s
\\
&=& \exp\left(\theta_{B,k,l}\, \dsum_{i<j}^n x_{i,j}\, z_{i,k}\, z_{j,l}\right),
\ee
where
\beno
\theta_{B,k,l}
&=& \logit(\mu_{k,l}) 
&\in& \mR,
&& l < k = 1, \dots, K.
\ee
As a consequence,
the joint probability mass function of random graph $\bX$ is proportional to
\beno
p_{\bta(\btheta, \bz)}(\bx)
&\propto& \dprod_{k=1}^K \exp\left(\theta_{W,k,k}\, \dsum_{i<j}^n x_{i,j}\, z_{i,k}\, z_{j,k}\right)\; \dprod_{l=1}^{k-1} \exp\left(\theta_{B,k,l}\, \dsum_{i<j}^n x_{i,j}\, z_{i,k}\, z_{j,l}\right).
\ee
However,
while stochastic block models are popular,
the assumption that edges within and between blocks are independent is restrictive,
because edges among nodes that are close---in the sense of being members of the same block---may very well be dependent.}

\s

\alert{{\bf Example 2. Model with local dependence.}
To allow edges within blocks to be dependent,
we build on the stochastic block model described above,
but tilt the within-block probability mass function of the stochastic block model as follows:
\beno
p_{\bta_{W,k}(\btheta, \bz)}(\bx_{k,k})
\;\propto\; \exp\left(\theta_{W,k,k,1}\, \dsum_{i<j}^n x_{i,j}\, z_{i,k}\, z_{j,k}\right)\; \exp\left(\theta_{W,k,k,2}\, \dsum_{i<j}^n x_{i,j}\, z_{i,k}\, z_{j,k}\, \one_{i,j}(\bx_{k,k})\right),
\ee
where $\theta_{W,k,k,1} \in \mR$ and $\theta_{W,k,k,2} \in \mR$ and $\one_{i,j}(\bx_{k,k})$ is an indicator function,
which is $1$ if nodes $i$ and $j$ are both connected to one or more other nodes in block $k$,
and is $0$ otherwise.
If $x_{i,j} = 1$ and $\one_{i,j}(\bx_{k,k}) = 1$,
the edge between nodes $i$ and $j$ is called transitive,
because $i$ and $j$ form transitive triples with other nodes,
which are known as triangles in the random graph literature.
The first term of the within-block probability mass function shown above corresponds to the within-block probability mass function of the stochastic block model.
The second term tilts the within-block probability mass function of the stochastic block model:
If $\theta_{W,k,k,2} > 0$,
the model rewards within-block subgraphs with transitive edges,
whereas $\theta_{W,k,k,2} < 0$ penalizes them,
and $\theta_{W,k,k,2} = 0$ neither rewards nor penalizes them---in which case the model reduces to the stochastic block model.}

\alert{Using the same between-block probability mass functions as the stochastic block model,
the joint probability mass function of random graph $\bX$ is proportional to
\beno
\label{example}
p_{\bta(\btheta, \bz)}(\bx)
&\propto& \dprod_{k=1}^K \exp\left(\theta_{W,k,k,1}\, \dsum_{i<j}^n x_{i,j}\, z_{i,k}\, z_{j,k}\right)\; \exp\left(\theta_{W,k,k,2}\, \dsum_{i<j}^n x_{i,j}\, z_{i,k}\, z_{j,k}\, \one_{i,j}(\bx_{k,k})\right)\s\s
\\
&\times& \dprod_{l=1}^{k-1} \exp\left(\theta_{B,k,l}\, \dsum_{i<j}^n x_{i,j}\, z_{i,k}\, z_{j,l}\right).
\ee}

\alert{The resulting model can capture an excess in the expected number of transitive edges within blocks, relative to the stochastic block model.
To demonstrate,
note that the resulting model is an exponential-family model and---by exponential-family theory \citep[][Corollary 2.5, p.\ 37]{Br86}---the expected number of transitive edges in block $\mA_k$ satisfies
\[
\begin{array}{cccccccccc}
\underbrace{\mbE_{\theta_{W,k,k,1},\, \theta_{W,k,k,2}>0}\; s_{k,k,2}(\bX)}
&>& \underbrace{\mbE_{\theta_{W,k,k,1},\, \theta_{W,k,k,2}=0}\; s_{k,k,2}(\bX)},
& k = 1, \dots, K,\\
\mbox{\em model with local dependence} && \mbox{\em stochastic block model}
\end{array}
\]
where $s_{k,k,2}(\bX)$ is the number of transitive edges in block $\mA_k$ and $\mbE_{\theta_{W,k,k,1},\, \theta_{W,k,k,2}}\, s_{k,k,2}(\bX)$ is the expectation of $s_{k,k,2}(\bX)$.
In other words,
the expected number of transitive edges in block $\mA_k$ is greater under the model with $\theta_{W,k,k,2} > 0$ than under the stochastic block model with $\theta_{W,k,k,2} = 0$,
assuming that both have the same edge parameters $\theta_{W,k,k,1}$ ($k = 1, \dots, K$).
Therefore,
the model with local dependence can capture an excess in the expected number of transitive edges within blocks,
relative to the stochastic block model.
In addition,
models with local dependence can capture excesses in the expected number of other subgraph statistics within blocks by adding suitable model terms.}
 
\paragraph{How models with within-block edge and transitive edge terms differ from the ``triangle model"}

It is worth noting that the model with local dependence induced by within-block edge and transitive edge terms differs in two important ways from the infamous triangle model with edge and triangle terms,
which has been known to be ill-behaved since the 1980s \citep{St86,Jo99,Sc09b,ChDi11}:
\bi
\item The model with local dependence restricts dependence among edges to subsets of nodes,
i.e.,
blocks.
As long as the blocks are not too large,
the overall dependence induced by the model is weak.
By contrast, 
the triangle model does not restrict dependence to subsets of nodes,
leading to undesirable behavior in large graphs.
\item Within blocks,
the model with local dependence and positive within-block transitive edge parameters assumes that,
for each pair of nodes,
the value added by the first triangle to the log odds of the conditional probability of an edge is positive,
but additional triangles add nothing.
By contrast,
the triangle model assumes that,
for each pair of nodes,
each additional triangle has the same added value,
which is unreasonable and results in undesirable behavior in large graphs.
\ei
These sensible assumptions ensure that models with local dependence have more desirable properties than the triangle model \citep{ScSt16}.

\paragraph{Exponential-family representations of models}

It is worth noting that Models 1 and 2 can be represented in exponential-family form:
\beno
p_{\bta(\btheta, \bz)} (\bx)
&=&\exp(\langle\bta(\btheta, \bz),\, s(\bx)\rangle - \psi(\bta(\btheta, \bz))),
& \bx \in \mbX,
\ee
where $\langle\bta(\btheta, \bz),\, s(\bx)\rangle$ denotes the inner product of a vector of natural parameters $\bta: \bTheta \times \mZ \mapsto \mR^d$ and a vector of sufficient statistics $s: \mbX \mapsto \mR^d$,
and $\psi(\bta(\btheta, \bz))$ ensures that $p_{\bta(\btheta, \bz)} (\bx)$ sums to $1$.
While exponential-family representations are not needed to introduce Models 1 and 2,
the properties of exponential families facilitate the theoretical results in Section \ref{sec:theory},
concerned with approximations of likelihood functions.

\section{Likelihood-based inference}
\label{sec:approximation}

While it is tempting to base statistical inference for the unknown block structure $\bz$ and the unknown parameter vector $\btheta$ on the likelihood function,
likelihood-based inference for models with local dependence is challenging.
The main reason is that the probability mass function $p_{\bta(\btheta, \bz)}(\bx)$ is intractable,
because the within-block probability mass functions $p_{\bta_{W,k}(\btheta, \bz)}(\bx_{k,k})$ are intractable ($k = 1, \dots, K$).
The intractability of $p_{\bta_{W,k}(\btheta, \bz)}(\bx_{k,k})$ is rooted in the fact that its normalizing constant is a sum over all $\exp(\binom{|\mA_k|}{2} \log 2)$ possible within-block subgraphs of block $\mA_k$,
which cannot be computed unless $\mA_k$ is small,
i.e.,
unless $|\mA_k| \ll 10$ ($k = 1, \dots, K$).

To facilitate likelihood-based inference,
we introduce tractable approximations of the intractable probability mass function $p_{\bta(\btheta, \bz)}(\bx)$ in Section \ref{sec:idea} and support them by theoretical results in Section \ref{sec:theory}.
A statistical algorithm that takes advantage of such approximations is introduced in Section \ref{sec:variational}.

\subsection{Approximate likelihood functions: motivation}
\label{sec:idea}

Suppose that we want to estimate both $\bz$ and $\btheta$.
It is natural to estimate them by using an iterative algorithm that cycles through updates of $\bz$ and $\btheta$ as follows:
\bi
\item[] Step 1: Update $\bz$ given $\btheta$.
\item[] Step 2: Update $\btheta$ given $\bz$.
\ei
The algorithm sketched above is generic and cannot be used in practice,
but regardless of which specific algorithm is used---whether EM, Monte Carlo EM, variational EM, Bayesian Markov chain Monte Carlo, or other algorithms---most of them have in common that Step 1 is either infeasible or time-consuming,
whereas Step 2 is less problematic than Step 1.

\paragraph{Step 1}
Step 1 is either infeasible or time-consuming,
because the probability mass function $p_{\bta(\btheta, \bz)}(\bx)$ is intractable.
To demonstrate,
consider a Bayesian Markov chain Monte Carlo algorithm that updates $\bz = (\bz_1, \dots, \bz_n)$ given $\btheta$ by Gibbs sampling.
Gibbs sampling of $\bz_1, \dots, \bz_n$ turns out to be infeasible,
because the full conditional distributions of $\bz_1, \dots, \bz_n$ depend on the intractable within-block probability mass functions $p_{\bta_{W,k}(\btheta, \bz)}(\bx_{k,k})$ ($k = 1, \dots, K$).
One could approximate them by Monte Carlo samples of within-block subgraphs,
but such approximations may not generate Markov chain Monte Carlo samples from the target distribution \citep{LiJi16} and are problematic on computational grounds:
\bi
\item Using Monte Carlo approximations of within-block probability mass functions is infeasible when the number of nodes $n$ is large,
because such approximations are needed for each update of each of the $n$ block memberships $\bz_1, \dots, \bz_n$.
\item Worse, 
the $n$ block memberships $\bz_1, \dots, \bz_n$ cannot be updated in parallel,
because the block membership of one node depends on the block memberships of other nodes.
\ei
As a consequence,
Step 1 is infeasible when $n$ is large.

\paragraph{Step 2}
Step 2 is less problematic than Step 1.
While the probability mass function $p_{\bta(\btheta, \bz)}(\bx)$ is intractable and may have to be approximated by Monte Carlo methods \citep{HuHa04},
such Monte Carlo approximations are needed once to update $\btheta$ given $\bz_1, \dots, \bz_n$,
whereas Monte Carlo approximations are needed $n$ times to update $\bz_1, \dots, \bz_n$ given $\btheta$ one by one.
In addition,
the probability mass function $p_{\bta(\btheta, \bz)}(\bx)$ decomposes into between- and within-block probability mass functions $p_{\bta(\btheta, \bz)}(\bx_{k,l})$ ($k \leq l = 1, \dots, K$) and hence within-block probability mass functions can be approximated in parallel,
i.e.,
on multi-core computers or computing clusters.

\paragraph{Approximations}
To enable feasible updates of $\bz$ given $\btheta$ when $n$ is large,
we are interested in approximating the intractable probability mass function $p_{\bta(\btheta, \bz)}(\bx)$ by a tractable probability mass function.
To do so,
we focus on models with between-block edge terms $\theta_{B,k,l} \sum_{i<j}^n x_{i,j}\, z_{i,k}\, z_{j,l}$ and within-block edge terms $\theta_{W,k,k,1} \sum_{i<j}^n x_{i,j}\, z_{i,k}\, z_{j,k}$ along with additional within-block terms that induce local dependence within blocks.
We denote the vector of between- and within-block edge parameters by $\btheta_1$ and the vector of all other within-block parameters by $\btheta_2$.
We assume that $\btheta_2 = \bm{0}$ reduces the model to the stochastic block model.
An example is given by the model with between-block edge terms and within-block edge and transitive edge terms described in Section \ref{sec:hergm}:
the parameter vector $\btheta_1$ consists of the between-block edge parameters $\theta_{B,k,l}$ ($l < k = 1, \dots, K$) and the within-block edge parameters $\theta_{W,k,k,1}$ ($k = 1, \dots, K$),
whereas the parameter vector $\btheta_2$ consists of the within-block transitive edge parameters $\theta_{W,k,k,2}$ ($k = 1, \dots, K$).
If the parameter vector $\btheta_2$ is set to $\bm{0}$,
the model reduces to the stochastic block model,
under which edges within and between blocks are independent.

Such models have two useful properties:
\bi
\item The probability mass functions $p_{\bta(\btheta, \bz)}(\bx)$ and $p_{\bta(\btheta_1, \btheta_2=\bm{0}, \bz)}(\bx)$ impose the same probability law on between-block subgraphs.
\item The probability mass function $p_{\bta(\btheta_1, \btheta_2=\bm{0}, \bz)}(\bx)$ is tractable,
because edges between and within blocks are independent under all $\bz$.
\ei
We henceforth approximate $p_{\bta(\btheta, \bz)}(\bx)$ by $p_{\bta(\btheta_1, \btheta_2=\bm{0}, \bz)}(\bx)$,
which corresponds to the probability mass function of the stochastic block model.

The idea underlying the approximation is that when the blocks are not too large,
most pairs of nodes are not members of the same block.
Since $p_{\bta(\btheta, \bz)}(\bx)$ and $p_{\bta(\btheta_1, \btheta_2=\bm{0}, \bz)}(\bx)$ impose the same probability law on possible edges between pairs of nodes that are not members of the same block,
$p_{\bta(\btheta, \bz)}(\bx)$ and $p_{\bta(\btheta_1, \btheta_2=\bm{0}, \bz)}(\bx)$ agree on most of the random graph.
Therefore,
$p_{\bta(\btheta, \bz)}(\bx)$ can be approximated by $p_{\bta(\btheta_1, \btheta_2=\bm{0}, \bz)}(\bx)$ for the purpose of updating $\bz$ given $\btheta$.
Suppose, 
e.g.,
that we consider to update $\bz$ given $\btheta$ by replacing $\bz$ by $\bz^\prime \neq \bz$.
We may decide to do so if the loglikelihood ratio
\beno
\log \dfrac{p_{\bta(\btheta, \bz^\prime)}(\bx)}{p_{\bta(\btheta, \bz)}(\bx)}
\= \log p_{\bta(\btheta, \bz^\prime)}(\bx) - \log p_{\bta(\btheta, \bz)}(\bx)
\ee
is large:
e.g.,
the acceptance probability of Metropolis-Hastings algorithms depends on the loglikelihood ratio above.
If $p_{\bta(\btheta, \bz)}(\bx)$ can be approximated by $p_{\bta(\btheta_1, \btheta_2=\bm{0}, \bz)}(\bx)$,
we can base the decision on $\log p_{\bta(\btheta_1, \btheta_2=\bm{0}, \bz^\prime)}(\bx) - \log p_{\bta(\btheta_1, \btheta_2=\bm{0}, \bz)}(\bx)$ rather than $\log p_{\bta(\btheta, \bz^\prime)}(\bx) - \log p_{\bta(\btheta, \bz)}(\bx)$,
because
\beno
&& \log p_{\bta(\btheta, \bz^\prime)}(\bx) - \log p_{\bta(\btheta, \bz)}(\bx)
\= [\log p_{\bta(\btheta_1, \btheta_2=\bm{0}, \bz^\prime)}(\bx) - \log p_{\bta(\btheta_1, \btheta_2=\bm{0}, \bz)}(\bx)]\s
\\
&+& [\log p_{\bta(\btheta, \bz^\prime)}(\bx) - \log p_{\bta(\btheta_1, \btheta_2=\bm{0}, \bz^\prime)}(\bx)]
&-& [\log p_{\bta(\btheta, \bz)}(\bx) - \log p_{\bta(\btheta_1, \btheta_2=\bm{0}, \bz)}(\bx)].
\ee
Therefore,
as long as
\beno
\max\limits_{\bz}\, |\log p_{\bta(\btheta, \bz)}(\bx) - \log p_{\bta(\btheta_1, \btheta_2=\bm{0}, \bz)}(\bx)|
\ee
is small,
we have
\beno
\log p_{\bta(\btheta, \bz^\prime)}(\bx) - \log p_{\bta(\btheta, \bz)}(\bx)
&\approx& \log p_{\bta(\btheta_1, \btheta_2=\bm{0}, \bz^\prime)}(\bx) - \log p_{\bta(\btheta_1, \btheta_2=\bm{0}, \bz)}(\bx).
\ee

The advantage of approximating $p_{\bta(\btheta, \bz)}(\bx)$ by $p_{\bta(\btheta_1, \btheta_2=\bm{0}, \bz)}(\bx)$ is that there exist methods for stochastic block models to estimate the block structure from large random graphs \citep[e.g.,][]{DaPiRo08,RoChYu11,AmChBiLe13,VuHuSc12}.
We take advantage of such methods in Section \ref{sec:variational},
but we first shed light on the conditions under which $\max_{\bz}\, |\log p_{\bta(\btheta, \bz)}(\bx) - \log p_{\bta(\btheta_1, \btheta_2=\bm{0}, \bz)}(\bx)|$ is small.
 
\subsection{Approximate likelihood functions: theoretical results}
\label{sec:theory}

We show that updates of $\bz$ given $\btheta$ can be based on $p_{\bta(\btheta_1, \btheta_2=\bm{0}, \bz)}(\bx)$ rather than $p_{\bta(\btheta, \bz)}(\bx)$ by showing that
\beno
\max\limits_{\bz}\, |\log p_{\bta(\btheta, \bz)}(\bX) - \log p_{\bta(\btheta_1, \btheta_2=\bm{0}, \bz)}(\bX)|
\ee
is small with high probability when the blocks are not too large.

We start with a special case in Theorem \ref{t.concentration2} and then present more general results in Theorem \ref{t.concentration}.
\alert{To prepare the ground for Theorems \ref{t.concentration2} and \ref{t.concentration},
let $\mZ = \{1, \dots, K\}^n$ be the set of all block structures,
and denote by $m(\bz)$ the size of the largest block under $\bz \in \mbZ$.
Let $\mbS \subseteq \mZ$ be any subset of block structures that includes the data-generating block structure $\zs \in \mbZ$,
and denote by $m(\mbS)$ the size of the largest block among all block structures $\bz \in \mbS$,
so that $m(\bz) \leq m(\mbS)$ for all $\bz \in \mbS$.
We allow the size of the largest block $m(\mbS)$ to increase as a function of the number of nodes $n$:
e.g.,
$m(\mbS)$ may be a constant multiple of $\log n$ or $n^{\alpha}$ ($\alpha < 1$).
The size of the largest data-generating block is denoted by $\norm{\mA}_\infty = \max_{1 \leq k \leq K} |\mA_k|$.} 

\vspace{-.25cm}

\begin{theorem}
\label{t.concentration2}
Consider the model with between-block edge terms and within-block edge and transitive edge terms described in Section \ref{sec:hergm}.
Let $\bTheta = \prod_{k\leq l}^K \bTheta_{k,l}$ be the parameter space,
$\bTheta_{k,k}$ be a compact subset of $\mbR^2$ ($k = 1, \dots, K$),
and $\bTheta_{k,l}$ be a compact subset of $\mbR$ ($l < k = 1, \dots, K$).
Choose $\epsilon \in (0, 1)$ as small as desired.
Then there exists a universal constant $c > 0$ such that,
for all $\btheta\in\bTheta$,
with at least probability $1 - \epsilon$,
\beno
\max\limits_{\bz\in\mbS} |\log p_{\bta(\btheta, \bz)}(\bX) - \log p_{\bta(\btheta_1, \btheta_2=\bm{0}, \bz)}(\bX)|\, <\, 2\, c\, \sqrt{-\log\dfrac{\epsilon}{2} + n \log K}\, \sqrt{K}\, m(\mbS)^2\, \norm{\mA}_\infty^2\, \log n.
\ee
\end{theorem}

The proof of Theorem \ref{t.concentration2} can be found in the supplement.
The basic idea underlying Theorem \ref{t.concentration2} is that the deviation $\max_{\bz} |\log p_{\bta(\btheta, \bz)}(\bx) - \log p_{\bta(\btheta_1, \btheta_2=\bm{0}, \bz)}(\bx)|$ cannot be too large when the blocks are not too large,
because most pairs of nodes are not members of the same block and $p_{\bta(\btheta, \bz)}(\bx)$ and $p_{\bta(\btheta_1, \btheta_2=\bm{0}, \bz)}(\bx)$ impose the same probability law on possible edges between pairs of nodes that are not members of the same block.

\alert{Theorem \ref{t.concentration2} implies that the largest deviation of the loglikelihood function under the unrestricted model and the restricted model is smaller than $n^2$ with high probability,
provided the blocks are not too large.
To see that,
observe that $\norm{\mA}_\infty \leq m(\mbS)$ implies
\beno
2\, c\, \sqrt{-\log\dfrac{\epsilon}{2} + n \log K}\, \sqrt{K}\, m(\mbS)^2\, \norm{\mA}_\infty^2\, \log n
\hide{
\lte 2\, c\, \sqrt{-\log\dfrac{\epsilon}{2} + n \log n}\, \sqrt{n}\, m(\mbS)^2\, \norm{\mA}_\infty^2\, \log n\s
\\
\lte 2\, c\, \sqrt{-2\, \log\dfrac{\epsilon}{2}\, n \log n}\; \sqrt{n}\; m(\mbS)^2\; \norm{\mA}_\infty^2\, \log n\s
\\
\lte 2\; c\; \sqrt{-2\, \log\dfrac{\epsilon}{2}}\; n\; \norm{\mA}_\infty^2\; m(\mbS)^2\; (\log n)^{3/2}\s
\\
\lte \delta(\epsilon)\; n\; \norm{\mA}_\infty^2\; m(\mbS)^2\; (\log n)^{3/2}
}
\lte \delta(\epsilon)\; n\; m(\mbS)^4\; (\log n)^{3/2},
\ee
where $\delta(\epsilon) > 0$ is a function of $\epsilon$ but is not a function of the number of nodes $n$.
In other words,
as long as the size $m(\mbS)$ of the largest block in $\mbS$ satisfies $m(\mbS) \ll n^{1/4}\, /\, (\log n)^{3/8}$,
the maximum deviation is much smaller than $n^2$ with high probability.
As a consequence,
the largest deviation of the loglikelihood function under the unrestricted model and the restricted model is small with high probability---note that in dense random graphs many quantities are of order $n^2$,
so quantities of order less than $n^2$ may be considered small.
For example,
consider Bernoulli$(\mu)$ random graphs,
under which possible edges $X_{i,j}$ are independent Bernoulli$(\mu)$ random variables \citep{ErRe59,ErRe60}.
If a Bernoulli$(\mu)$ random graph is dense in the sense that $\mu = \mbE\, X_{i,j} \in (0, 1)$ does not decrease as a function of the number of nodes $n$,
then the expected number of edges is of order $n^2$,
\beno
\mbE\, \dsum_{i<j}^n X_{i,j} 
\= \dis\binom{n}{2}\, \mu,
\ee
and so is the expected loglikelihood function of the natural parameter $\theta=\logit(\mu) \in \mR$, 
\beno
\mbE\, \log p_{\theta}(\bX)
\= \dis\binom{n}{2}\, \left(\theta\; \mu - \log(1 + \exp(\theta))\right).
\ee
Other quantities are likewise of order $n^2$ in dense random graphs,
so quantities of order less than $n^2$ may be considered small.
Thus,
the largest deviation of the loglikelihood function under the unrestricted model and the restricted model is small with high probability.}

\alert{Last, 
but not least, 
note that $m(\mbS) \ll n^{1/4}\, /\, (\log n)^{3/8}$ implies that the number of blocks $K$ must satisfy $K\, \gg\, n^{3/4}\, (\log n)^{3/8}$.
In other words,
the number of blocks $K$ needs to grow as function of the number of nodes $n$.
It is worth noting that in the special case of stochastic block models,
it is known that $K$ is allowed to grow as fast as $n\, / \log n$ in dense-graph settings \citep{ZhZh16}.} 

It turns out that the result in Theorem \ref{t.concentration2} is not limited to the model with between-block edge terms and within-block edge and transitive edge terms,
but is a special case of more general results.
To introduce these more general results,
we make the following assumptions.
We assume that the probability mass function $p_{\bta(\btheta, \bz)} (\bx)$ can be represented in exponential-family form,
\beno
p_{\bta(\btheta, \bz)} (\bx)
&=&\exp(\langle\bta(\btheta, \bz),\, s(\bx)\rangle - \psi(\bta(\btheta, \bz))),
& \bx \in \mbX,
\ee
and satisfies local dependence as defined in Section \ref{sec:hergm};
note that all models used in our paper can be represented in exponential-family form.
We assume that $\bta: \bTheta \times \mbZ \mapsto \bm{\Xi}$ and that $\bm{\Xi} \subseteq \mbox{int}(\mbN)$ is a subset of the interior $\mbox{int}(\mbN)$ of the natural parameter space $\mbN$ of the exponential family \citep{Br86}.
Let $\mbE \equiv \mbE_{\bta^\star}$ be the expectation under the data-generating parameter vector $\bta^\star \equiv \bta(\bthetas, \zs)$,
where $(\bthetas, \zs) \in \bTheta \times \mbZ$ denotes the data-generating value of $(\btheta, \bz) \in \bTheta \times \mbZ$.
We denote by $d: \mbX \times \mbX \mapsto \{0, 1, 2, \dots\}$ the Hamming metric,
which is defined by
\beno
d(\bx_1, \bx_2)
\= \dsum_{i<j}^n \one_{x_{1,i,j} \neq x_{2,i,j}},
& (\bx_1, \bx_2) \in \mX \times \mX,
\ee
where $\one_{x_{1,i,j} \neq x_{2,i,j}}$ is $1$ if $x_{1,i,j} \neq x_{2,i,j}$ and is $0$ otherwise.
The main assumptions can then be stated as follows.
\bi
\item[[C.1]\hspace{-.15cm}] There exists $c>0$ such that, 
for all $(\btheta, \bz) \in \bTheta \times \mZ$ and all $(\bx_1,\bx_2) \in \mbX \times \mbX$,
\beno
\left|\left\langle\bta(\btheta, \bz),\, s(\bx_1) - s(\bx_2)\right\rangle\right|
&\leq& c\; d(\bx_1, \bx_2)\; m(\bz)\, \log n.
\ee
\item[[C.2]\hspace{-.15cm}] There exists $c>0$ such that, 
for all $(\btheta_{k,l,1},\, \btheta_{k,l,2}) \in \bTheta_{k,l} \times \bTheta_{k,l}$ and all $\bz \in \mbZ$,
\beno
\left|\left\langle\bta_{k,l}(\btheta_{k,l,1}, \bz) - \bta_{k,l}(\btheta_{k,l,2}, \bz),\, \mbE_{\bta(\btheta, \bz)}\, s_{k,l}(\bX)\right\rangle\right|
&\leq& c\; \norm{\btheta_{k,l,1} - \btheta_{k,l,2}}\; m(\bz)^2\, \log n,
\ee
where $\bta_{k,l}(\btheta_{k,l}, \bz)$, $\btheta_{k,l}$, and $s_{k,l}(\bx)$ denote the subvectors of $\bta(\btheta, \bz)$, $\btheta$, and $s(\bx)$ corresponding to the subgraph between blocks $k$ and $l$ ($l < k$) or the subgraph of block $k$ ($k = l$)
and $\bTheta_{k,l}$ is a compact subset of $\mbR^{\rrr_{k,l}}$ 
($k \leq l = 1, \dots, K$).
\ei
\alert{Conditions [C.1] and [C.2] are smoothness conditions:
[C.1] states that the inner product of natural parameters and sufficient statistics must not be too sensitive to changes of edges,
whereas [C.2] states that the inner product of between- and within-block natural parameters and expected sufficient statistics must not be too sensitive to changes of parameters.}

\alert{An example of a model violating condition [C.1] is a model containing a sufficient statistic $s(\bx)$ of the form
\beno
s(\bx)
\= 
\begin{cases}
0 & \mbox{ if } \sum_{i<j}^n x_{i,j} = 0\s
\\
\binom{n}{2} & \mbox{ if } \sum_{i<j}^n x_{i,j} \in \{1, \dots, \binom{n}{2}\}.
\end{cases}
\ee 
Under such models,
adding a single edge can change the inner product of natural parameters and sufficient statistics by a multiple of $n^2$.
That would violate [C.1],
because [C.1] assumes that changing a single edge changes the inner product by at most $c\; m(\bz)\, \log n \leq c\; n \log n$.}
 
\alert{An example of a model violating condition [C.2] is a model with within-block sufficient statistics that count the number of triangles within blocks,
\beno
s_{k,k}(\bx)
\= \dsum_{h<i<j:\; h,\, i,\, j\, \in\, \mA_k}^n x_{i,h}\, x_{j,h}\, x_{i,j},
& k = 1, \dots, K,
\ee
with within-block parameters of the form $\eta_{k,k}(\btheta, \bz) = \theta_{k,k} \in \mR$ ($k = 1, \dots, K$).
If all nodes belong to block $\mA_1$,
then the inner product of the block's within-block natural parameter and expected sufficient statistic vector is a multiple of $|\theta_{1,1} - \theta_{1,1}^\prime|\, n^3$,
where $\theta_{1,1}$ and $\theta_{1,1}^\prime$ are two possible values of the within-block triangle parameter $\theta_{1,1}$ of block $\mA_1$.
As a result,
[C.2] would be violated,
because [C.2] requires the inner product to be at most $c\; |\theta_{1,1} - \theta_{1,1}^\prime|\; m(\bz)^2\, \log n\, \leq\, c\; |\theta_{1,1} - \theta_{1,1}^\prime|\; n^2\, \log n$.}

\alert{However,
the fact that some model specifications violate conditions [C.1] and [C.2] is not necessarily a major concern,
for two reasons.
First,
the specifications that violate conditions [C.1] and [C.2] are, 
more often than not, 
of limited interest in applications:
e.g.,
we are not aware of any good reason for using the sufficient statistic in the first example,
and the triangle term in the second example is known to induce model near-degeneracy in large networks and may therefore not be useful in practice \citep{St86,Jo99,Sc09b,ChDi11}.
Second,
while some ill-posed specifications do not satisfy conditions [C.1] and [C.2],
there are many well-posed specifications that do satisfy them:
e.g.,
conditions [C.1] and [C.2] are satisfied by the model with between-block edge and within-block edge and transitive edge terms,
which we verify in the proof of Theorem \ref{t.concentration2}.
In addition,
conditions [C.1] and [C.2] cover the models with size-dependent parameterizations used in Sections \ref{sec:sim_full} and \ref{sec:application}.}

The following result,
Theorem \ref{t.concentration},
is a generalization of Theorem \ref{t.concentration2}.

\begin{theorem}
\label{t.concentration}
Consider a model with local dependence satisfying conditions [C.1] and [C.2].
Choose $\epsilon \in (0, 1)$ as small as desired. 
Then there exists a universal constant $c > 0$ such that,
for all $\btheta\in\bTheta$,
with at least probability $1 - \epsilon$,
\beno
\max\limits_{\bz\in\mbS} |\log p_{\bta(\btheta, \bz)}(\bX) - \log p_{\bta(\btheta_1, \btheta_2=\bm{0}, \bz)}(\bX)|\, <\, 2\, c\, \sqrt{-\log\dfrac{\epsilon}{2} + n \log K}\, \sqrt{K}\, m(\mbS)^2\, \norm{\mA}_\infty^2\, \log n.
\ee
\end{theorem}

The proof of Theorem \ref{t.concentration} can be found in the supplement.
An application of Theorem \ref{t.concentration} to the model with between-block edge terms and within-block edge and transitive edge terms can be found in Theorem \ref{t.concentration2}.
\alert{\paragraph{Trade-off between $m(\mS)$ and the recovery of block structure} 
It is worth noting that the upper bound $m(\mS)$ on the sizes of blocks cannot be too small,
because it is not possible to recover the block structure with high probability when the blocks are too small.
However,
\citet{ZhZh16} showed that under stochastic block models the number of blocks $K$ is allowed to grow as fast as $n\, / \log n$ in dense-graph settings,
which implies that the sizes of the blocks can be as small as $\log n$.
These important issues are studied in more depth in, e.g., \citet{ZhZh16} and \citet{Gaetal17}.}

\section{Two-step estimation approach}
\label{sec:variational}

We propose a two-step estimation approach that takes advantage of the theoretical results of Section \ref{sec:approximation} and enables large-scale estimation of models with local dependence.

To describe the two-step estimation approach,
assume that $\bz = (\bz_1, \dots, \bz_n)$ is the observed value of a random variable $\bZ = (\bZ_1, \dots, \bZ_n)$ with distribution
\beno
\bZ_i &\iid& \mbox{Multinomial}(1,\, \bpi = (\pi_1, \dots, \pi_K)),
& i = 1, \dots, n.
\ee
It is natural to base statistical inference on the observed-data likelihood function
\beno
\mL(\btheta, \bpi)
&=& \dsum_{\bz \in \mbZ} p_{\bta(\btheta, \bz)}(\bx)\, p_{\bpi}(\bz).
\ee
The problem is that $\mL(\btheta, \bpi)$ is intractable,
because $p_{\bta(\btheta, \bz)}(\bx)$ is intractable and the set $\mbZ$ contains $\exp(n \log K)$ elements. 

The first problem can be solved by taking advantage of the theoretical results of Section \ref{sec:approximation},
which suggest that $p_{\bta(\btheta, \bz)}(\bx)$ can be approximated by $p_{\bta(\btheta_1, \btheta_2=\bm{0}, \bz)}(\bx)$ provided that the blocks are not too large. 
\alert{A complication is that $p_{\bta(\btheta, \bz)}(\bx)$ and $p_{\bta(\btheta_1, \btheta_2=\bm{0}, \bz)}(\bx)$ may not be close when the block structure $\bz \in \mZ\, \setminus\, \mbS$ is not contained in $\mbS$,
in which case some of the blocks can be large and the within-block models of large blocks can induce strong dependence.
In the worst case, 
all nodes belong to a single block, 
in which case $p_{\bta(\btheta, \bz)}(\bx)$ can induce strong dependence throughout the random graph while $p_{\bta(\btheta_1, \btheta_2=\bm{0}, \bz)}(\bx)$ induces no dependence at all.}
However,
the basic inequality
\beno
\dsum_{\bz \in \mbS} p_{\bta(\btheta, \bz)}(\bx)\, p_{\bpi}(\bz)
\lte \mL(\btheta, \bpi)
&\leq& \dsum_{\bz \in \mbS} p_{\bta(\btheta, \bz)}(\bx)\, p_{\bpi}(\bz) + \mbP_{\bpi}(\bZ \in \mbZ \setminus \mbS)
\ee
suggests that as long as the event $\bZ \in \mbZ \setminus \mbS$ is a rare event in the sense that $\mbP_{\bpi}(\bZ \in \mbZ \setminus \mbS)$ is close to $0$,
$\mL(\btheta, \bpi)$ can be approximated by $\mL(\btheta_1, \btheta_2\hspace{-.1cm}=\hspace{-.1cm}\bm{0}, \bpi)$:
\beno
\mL(\btheta, \bpi)
&=& \dsum_{\bz \in \mbZ} p_{\bta(\btheta, \bz)}(\bx)\, p_{\bpi}(\bz)
&\approx& \dsum_{\bz \in \mbS} p_{\bta(\btheta, \bz)}(\bx)\, p_{\bpi}(\bz)\s
\\
&\approx& \dsum_{\bz \in \mbS} p_{\bta(\btheta_1, \btheta_2=\bm{0}, \bz)}(\bx)\, p_{\bpi}(\bz)
&\approx& \dsum_{\bz \in \mbZ} p_{\bta(\btheta_1, \btheta_2=\bm{0}, \bz)}(\bx)\, p_{\bpi}(\bz)
&=& \mL(\btheta_1, \btheta_2\hspace{-.1cm}=\hspace{-.1cm}\bm{0}, \bpi).
\ee
The assumption that $\bZ \in \mbZ \setminus \mbS$ is a rare event---i.e., the probabilities $\pi_1, \dots, \pi_K$ are small---makes sense in a wide range of applications,
because communities in real-world networks tend to be small \citep[see, e.g., the discussion of][]{RoChYu11}. 
Therefore,
as long as $\bZ \in \mbZ \setminus \mbS$ is a rare event,
we can base statistical inference concerning the block structure on $\mL(\btheta_1, \btheta_2\hspace{-.1cm}=\hspace{-.1cm}\bm{0}, \bpi)$ rather than $\mL(\btheta, \bpi)$.
To simplify the notation, 
we write henceforth $\mL(\btheta_1, \bpi)$ instead of $\mL(\btheta_1, \btheta_2\hspace{-.1cm}=\hspace{-.1cm}\bm{0}, \bpi)$.

The second problem can be solved by methods developed for stochastic block models,
because $\mL(\btheta_1, \bpi)$ is the observed-data likelihood function of a stochastic block model.
There are many stochastic block model methods that could be used,
such as profile likelihood \citep{BiCh09}, 
pseudo-likelihood \citep{AmChBiLe13},
spectral clustering \citep{RoChYu11}, 
and variational methods \citep{DaPiRo08,VuHuSc12}.
Among these methods,
we found that the variational methods of \citet{VuHuSc12} work best in practice.
In addition,
the variational methods of \citet{VuHuSc12} have the advantage of being able to estimate stochastic block models from networks with hundreds of thousands of nodes due to a running time of $O(n)$ for sparse random graphs and $O(n^2)$ for dense random graphs \citep{VuHuSc12}.
Some consistency and asymptotic normality results for variational methods for stochastic block models were established by \citet{CeDaLa11} and \citet{BiChChZh13}.

Variational methods approximate $\ell(\btheta_1, \bpi) = \log\mL(\btheta_1, \bpi)$ by introducing an auxiliary distribution $a(\bz)$ with support $\mbZ$ and lower bound $\ell(\btheta_1, \bpi)$ by using Jensen's ineqality:
\beno
\label{eq:lower_bound}
\ell(\btheta_1, \bpi)
\= \log \dsum_{\bz \in \mbZ} a(\bz)\, \dfrac{p_{\bta(\btheta_1, \btheta_2=\bm{0}, \bz)}(\bx)\, p_{\bpi}(\bz)}{a(\bz)}\s
\\
\gte \dsum_{\bz \in \mbZ} a(\bz)\, \log \dfrac{p_{\bta(\btheta_1, \btheta_2=\bm{0}, \bz)}(\bx)\, p_{\bpi}(\bz)}{a(\bz)}
&\define& \hat\ell(\btheta_1, \bpi).
\ee
Each auxiliary distribution with support $\mbZ$ gives rise to a lower bound on $\ell(\btheta_1, \bpi)$.
To choose the best auxiliary distribution---i.e., the auxiliary distribution that gives rise to the tightest lower bound on $\ell(\btheta_1, \bpi)$---we choose a family of auxiliary distributions and select the best member of the family.
In practice,
an important consideration is that the resulting lower bound is tractable.
Therefore,
we confine attention to a family of auxiliary distributions under which the resulting lower bounds are tractable.
A natural choice is given by a family of  auxiliary distributions under which the block memberships are independent:
\beno
\bZ_i &\ind& \mbox{Multinomial}(1,\, \bm{\alpha}_i = (\alpha_{i,1}, \dots, \alpha_{i,K})),
& i = 1, \dots, n.
\ee
By the independence of block memberships under the auxiliary distribution,
we obtain the following tractable lower bound on $\ell(\btheta_1, \bpi)$:
\beno
\label{eq:approx_lb}
&& \hat\ell(\balpha; \btheta_1, \bpi) 
\;\define\; \dsum_{\bz \in \mbZ} a_{\balpha}(\bz)\, \log \dfrac{p_{\bta(\btheta_1, \btheta_2=\bm{0}, \bz)}(\bx)\, p_{\bpi}(\bz)}{a_{\balpha}(\bz)}\s
\\
\= \dsum_{i<j}^{n} \dsum_{k = 1}^{K} \dsum_{l = 1}^{K} \alpha_{i,k}\, \alpha_{j,l} \log p_{\bta(\btheta_1, \btheta_2=\bm{0}, z_{i,k}=1,\, z_{j,l}=1,\bz_{-i,j})}(x_{i,j})
+ \dsum_{i=1}^n \dsum_{k = 1}^{K} \alpha_{i,k}\, (\log \pi_k - \log \alpha_{i,k}),
\ee
where $p_{\bta(\btheta_1, \btheta_2=\bm{0}, z_{i,k}=1,\, z_{j,l}=1,\bz_{-i,j})}(x_{i,j})$ denotes the marginal probability mass function of $X_{i,j}$ and $\bz_{-i,j}$ denotes the block memberships of all nodes excluding nodes $i$ and $j$.

In practice,
we obtain the best lower bound on $\ell(\btheta_1, \bpi)$ by maximizing $\hat\ell(\balpha; \btheta_1, \bpi)$ with respect to $\balpha$.
Direct maximization of $\hat\ell(\balpha; \btheta_1, \bpi)$ with respect to $\balpha$ is possible but inconvenient,
because $\hat\ell(\balpha; \btheta_1, \bpi)$ contains products of $\alpha_{i,k}$ and $\alpha_{j,l}$.
As a consequence,
a fixed-point update of $\alpha_{i,k}$ depends on $(n - 1)\, K$ other terms $\alpha_{j,l}$ and hence fixed-point updates tend to be time-consuming and get stuck in local maxima,
as demonstrated by \citet{VuHuSc12}.
\citet{VuHuSc12} proposed an elegant approach to alleviating the problem by using minorization-maximization methods \citep{HuLa04}.
Such methods construct a minorizing function that approximates $\hat\ell(\balpha; \btheta_1, \bpi)$ but is easier to maximize than $\hat\ell(\balpha; \btheta_1, \bpi)$.
A function $\mQ(\balpha; \btheta_1, \bpi, \balpha^{(t)})$ of $\balpha$ minorizes $\hat\ell(\balpha; \btheta_1, \bpi)$ at point $\balpha^{(t)}$ at iteration $t$ of an iterative algorithm for maximizing $\hat\ell(\balpha; \btheta_1, \bpi)$ if
\beno
\mQ(\balpha; \btheta_1, \bpi, \balpha^{(t)}) \lte \hat\ell(\balpha; \btheta_1, \bpi) & \mbox{for all} & \balpha,\s
\\
\mQ(\balpha^{(t)}; \btheta_1, \bpi, \balpha^{(t)}) &=& \hat\ell(\balpha^{(t)}; \btheta_1, \bpi), 
\ee
where $\btheta_1, \bpi, \balpha^{(t)}$ are fixed.
In other words,
$\mQ(\balpha; \btheta_1, \bpi, \balpha^{(t)})$ is bounded above by $\hat\ell(\balpha; \btheta_1, \bpi)$ for all $\balpha$ and touches $\hat\ell(\balpha; \btheta_1, \bpi)$ at $\balpha = \balpha^{(t)}$.
As a result,
increasing $\mQ(\balpha; \btheta_1, \bpi, \balpha^{(t)})$ with respect to $\balpha$ increases $\hat\ell(\balpha; \btheta_1, \bpi)$.
\citet{VuHuSc12} showed that the following function minorizes $\hat\ell(\balpha; \btheta_1, \bpi)$ at point $\balpha^{(t)}$:
\beno
\label{eq:minorizer}
\mQ(\balpha; \btheta_1, \bpi, \balpha^{(t)}) 
&=& \dsum_{i<j}^n \dsum_{k = 1}^K \dsum_{l = 1}^K
\left(\alpha_{i,k}^2\; \dfrac{\alpha_{j,l}^{(t)}}{2\, \alpha_{i,k}^{(t)}} + \alpha_{j,l}^2\; \dfrac{\alpha_{i,k}^{(t)}}{2\, \alpha_{j,l}^{(t)}}\right) \log p_{\bta(\btheta_1, \btheta_2=\bm{0}, z_{i,k}=1,\, z_{j,l}=1,\bz_{-i,j})}(x_{i,j})\s 
\\
&+& \dsum_{i=1}^{n} \dsum_{k = 1}^{K} \alpha_{i,k} \left[\log \pi_k^{(t)} - \log \alpha_{i,k}^{(t)} + \left(1 - \dfrac{\alpha_{i,k}}{\alpha_{i,k}^{(t)}}\right)\right].
\ee
The minorizing function $\mQ(\balpha; \btheta_1, \bpi, \balpha^{(t)})$ is easier to maximize than $\hat\ell(\balpha; \btheta_1, \bpi)$,
because it replaces the products of $\alpha_{i,k}$ and $\alpha_{j,l}$ by sums of $\alpha_{i,k}^2$ and $\alpha_{j,l}^2$.
An additional advantage is that the maximization of $\mQ(\balpha; \btheta_1, \bpi, \balpha^{(t)})$ amounts to $n$ quadratic programming problems,
which can be solved in parallel.

\begin{table}[t]
\begin{center}
\fbox{
\begin{minipage}{6in}
\vspace{.125cm}
\begin{itemize}
\item[1.] Estimate $\bz$ along with $\bpi$ and $\btheta_1$ by iterating:
\begin{itemize}
\item[1.1] Update $\balpha$ by increasing $\mQ(\balpha; \btheta_1^{(t)}, \bpi^{(t)}, \balpha^{(t)})$ with respect to $\balpha_i$ subject to $\alpha_{i,k} \geq 0$ and $\sum_{k=1}^K \alpha_{i,k} = 1$ and denote the update by $\balpha_i^{(t+1)}$ ($i = 1, \dots, n$).
\item[1.2] Update $\bpi$ and $\btheta_1$ by maximizing $\hat\ell(\balpha^{(t+1)}; \btheta_1, \bpi)$ with respect to $\bpi$ and $\btheta_1$:
\bi
\item[---] Update $\pi_k^{(t + 1)} = (1/n)\, \sum_{i=1}^n \alpha_{i,k}^{(t + 1)}$,\;
$k = 1, \dots, K$.
\item[---] Update $\btheta_1^{(t+1)} = \argmax_{\btheta_1\in\bTheta_1} \hat\ell(\balpha^{(t+1)}; \btheta_1, \bpi^{(t+1)})$.
\ei
\end{itemize}
Upon convergence, 
we estimate the block memberhip indicators by $\widehat{z}_{i,k} = 1$ if $k = \argmax_{1 \leq l \leq K} \widehat{\alpha}_{i,l}$ and $\widehat{z}_{i,k} = 0$ otherwise ($i = 1, \dots, n$, $k = 1, \dots, K$),
where $\widehat{\balpha}$ denotes the final value of $\balpha$.
\item[2.] \alert{Estimate $\btheta$ given $\zh$ by $\bthetah = \argmax_{\btheta\in\bTheta} \hat\ell_{\zh}(\btheta)$ or $\mple = \argmax_{\btheta\in\bTheta} \tilde\ell_{\zh}(\btheta)$.}
\end{itemize}
\end{minipage}
}
\s
\caption{\label{tab:algorithm}Two-step estimation approach.}
\end{center}
\end{table}

We therefore propose a two-step estimation approach as described in Table \ref{tab:algorithm}. 
We discuss the two steps below and conclude with some comments on parallel computing.

\paragraph{Step 1}
The first step estimates $\bz$ based on $\balpha$.
We do so by increasing $\mQ(\balpha; \btheta_1, \bpi, \balpha^{(t)})$ with respect to $\balpha_i$ subject to the constraints $\alpha_{i,k} \geq 0$ and $\sum_{k=1}^K \alpha_{i,k} = 1$ ($i = 1, \dots, n$).
We increase rather than maximize $\mQ(\balpha; \btheta_1, \bpi, \balpha^{(t)})$,
because maximizing $\mQ(\balpha; \btheta_1, \bpi, \balpha^{(t)})$ is more time-consuming and algorithms maximizing $\mQ(\balpha; \btheta_1, \bpi, \balpha^{(t)})$ are more prone to end up in local maxima than algorithms increasing $\mQ(\balpha; \btheta_1, \bpi, \balpha^{(t)})$.
Since $\hat\ell(\balpha; \btheta_1, \bpi)$ and $\mQ(\balpha; \btheta_1, \bpi, \balpha^{(t)})$ depend on $\btheta_1$ and $\bpi$ and both are unknown,
we iterate between updates of $\balpha$ and updates of $\btheta_1$ and $\bpi$.
The updates of $\btheta_1$ and $\bpi$ are based on maximizing $\hat\ell(\balpha; \btheta_1, \bpi)$ with respect to $\btheta_1$ and $\bpi$ and are identical to the updates of \citet{VuHuSc12},
because $\btheta_2=\bm{0}$ reduces the model to a stochastic block model.
As a convergence criterion,
we use
\beno
\dfrac{|\hat\ell(\balpha^{(t+1)}; \btheta_1^{(t+1)}, \bpi^{(t+1)}) - \hat\ell(\balpha^{(t)}; \btheta_1^{(t)}, \bpi^{(t)})|}{\hat\ell(\balpha^{(t+1)}; \btheta_1^{(t+1)}, \bpi^{(t+1)})}
&<& \gamma,
\ee
where $\gamma > 0$ is a small constant.
Upon convergence, 
we estimate the block memberhip indicators by $\widehat{z}_{i,k} = 1$ if $k = \argmax_{1 \leq l \leq K} \widehat{\alpha}_{i,l}$ and $\widehat{z}_{i,k} = 0$ otherwise ($i = 1, \dots, n$, $k = 1, \dots, K$),
where $\widehat{\balpha}$ denotes the final value of $\balpha$.

\paragraph{Step 2}
We estimate $\btheta$ given $\widehat{\bz}$ by Monte Carlo maximum likelihood methods \citep{HuHa04} or maximum pseudolikelihood methods \citep{StIk90}.
Monte Carlo maximum likelihood methods exploit the fact that the loglikelihood function induced by $\zh$,
which is defined by
\beno
\label{identity}
\ell_{\zh}(\btheta)
&=& \log p_{\bta(\btheta, \zh)}(\bx) - \log p_{\bta(\btheta_0,\, \zh)}(\bx),
\ee
can be written as
\beno
\ell_{\zh}(\btheta)
&=& \langle\bta(\btheta, \widehat{\bz}) - \bta(\btheta_0,\, \zh),\, s(\bx)\rangle
- \log\mbE_{\bta(\btheta_0,\, \zh)} \exp(\langle\bta(\btheta, \widehat{\bz}) - \bta(\btheta_0,\, \zh),\, s(\bX)\rangle),
\ee
where $\btheta_0$ is a fixed parameter vector (e.g., $\btheta_0$ may be an educated guess of $\bthetas$).
In general,
the expectation $\mbE_{\bta(\btheta_0,\, \zh)}$ is intractable, 
but it can be estimated by a Monte Carlo sample average based on a Monte Carlo sample of graphs generated under $\bta(\btheta_0,\, \zh)$.
Therefore,
we can approximate $\ell_{\zh}(\btheta)$ by
\beno
\hat\ell_{\zh}(\btheta)
&=& \langle\bta(\btheta, \widehat{\bz}) - \bta(\btheta_0, \zh), s(\bx)\rangle - \log \widehat{\mbE}_{\bta(\btheta_0, \zh)} \exp(\langle\bta(\btheta, \widehat{\bz}) - \bta(\btheta_0, \zh), s(\bX)\rangle),
\ee
where $\widehat{\mbE}_{\bta(\btheta_0,\, \zh)}$ is a Monte Carlo approximation of $\mbE_{\bta(\btheta_0,\, \zh)}$ based on a Monte Carlo sample of graphs generated by using $\bta(\btheta_0,\, \zh)$.
Hence $\btheta$ given $\zh$ can be estimated by 
\beno
\label{mcmle}
\bthetah 
\= \argmax\limits_{\btheta\in\bTheta} \hat\ell_{\zh}(\btheta).
\ee
Additional details on Monte Carlo maximum likelihood methods can be found in \citet{HuHa04}.
We note that the local dependence of the model facilitates parallel computing,
which is discussed in the following paragraph.
Standard errors of $\bthetah$ can be based on the estimated Fisher information matrix,
although such standard errors are conditional on the estimated block structure $\zh$ and therefore do not reflect the uncertainty about $\zh$.
A parametric bootstrap approach would be an interesting alternative to capturing the additional uncertainty due to $\zh$,
but would be time-consuming.

\alert{
An alternative is to estimate $\btheta$ given $\zh$ by maximum pseudolikelihood estimators \citep{StIk90}.
Maximum pseudolikelihood estimators were invented by \citet{Bj74} to sidestep intractable likelihood functions of Markov random fields in spatial statistics,
and are known to be consistent estimators of Markov random fields with exponential-family parameterizations \citep[e.g.,][]{Co92}.
\citet{StIk90} adapted them to exponential-family random graph models.
While believed to be inferior to maximum likelihood estimators when the dependence among edges is strong and propagates throughout the random graph \citep[e.g.,][]{DuGiHa09},
maximum pseudolikelihood estimators seem to perform well when the dependence is weak,
e.g.,
when the dependence among edges is confined to non-overlapping or overlapping blocks \citep{StSc20}.
The maximum pseudolikelihood estimator is defined as 
\beno
\mple
\= \argmax\limits_{\btheta\in\bTheta} \tilde\ell_{\zh}(\btheta),
\ee
where 
\beno
\tilde\ell_{\zh}(\btheta)
\= \dsum_{i<j}^n \log p_{\bta(\btheta, \zh)}(x_{i,j} \mid \bx_{-\{i,j\}}).
\ee
Here, 
$\bx_{-\{i,j\}}$ denotes $\bx$ excluding $x_{i,j}$,
and $p_{\bta(\btheta, \zh)}(x_{i,j} \mid \bx_{-\{i,j\}})$ denotes the conditional probability of $X_{i,j} = x_{i,j}$ given $\bX_{-\{i,j\}} = \bx_{-\{i,j\}}$.
Since the conditional distributions of $X_{i,j}$ given the rest of the random graph are Bernoulli distributions,
computing maximum pseudolikelihood estimators consumes less time than computing Monte Carlo maximum likelihood estimators.
}

\paragraph{Parallel computing}
In Step 1,
the maximization of the minorizing function amounts to $n$ quadratic programming problems,
which can be solved in parallel.
In Step 2,
the local dependence induced by the model implies that the contributions of the between- and within-block subgraphs to the loglikelihood function and its gradient and Hessian can be computed in parallel.
Hence both steps can be implemented in parallel, 
which suggests that the two-step likelihood-based method can be used on a massive scale as long as the blocks are not too large and multi-core computers or computing clusters are available.

\section{Simulation results}
\label{sec:sim_full}

\alert{We assess the performance of the two-step estimation approach by conducting multiple simulation studies,
with the number of nodes $n$ ranging from 30 to 10,000,
the number of blocks $K$ ranging from 3 to 100,
and the sizes of blocks ranging from 5 to 100:
\bi
\item[I.] A small-scale simulation study to compare the block recovery of the gold standard, 
the Bayesian approach of \citet{ScHa13}, 
to the two-step estimation approach.
Since the Bayesian approach is too time-consuming to be applied to large networks,
we use small networks with $n=30$ nodes and $K=3$ blocks.
The $K = 3$ blocks consist of 10 nodes (balanced case) or 5, 10, 15 nodes (unbalanced case).
\item[II.] A large-scale simulation study to assess the block recovery in Step 1 of the two-step estimation approach,
using $K = 25, 50, 75, 100$ blocks of size $25, 50, 75, 100$.
\item[III.] A large-scale simulation study to assess how the block recovery in Step 1 of the two-step estimation approach is affected by the sparsity of between-block subgraphs,
using $K = 25$ blocks of size $25$.
\item[IV.] A large-scale simulation study to assess the parameter recovery in Step 2 of the two-step estimation approach,
using $K = 25, 50, 75, 100$ blocks of size $25, 50, 75, 100$.
\ei}

\begin{figure}[h]
\centering
\includegraphics[width = .95 \textwidth, keepaspectratio]{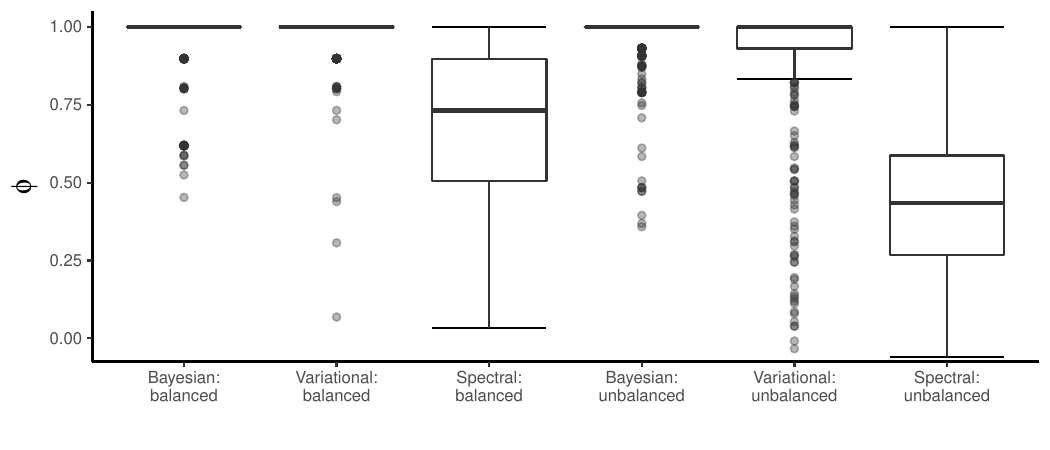}
\caption{Block recovery in terms of Yule's $\phi$-coefficient.
$500$ graphs with $n=30$ nodes and $K=3$ blocks with 10 nodes (balanced case) or 5, 10, 15 nodes (unbalanced case) were generated.
The figure compares the block recovery of the Bayesian approach,
the variational approach, 
and spectral clustering.
The Bayesian approach is the gold standard.
The variational approach is the default in Step 1 of the two-step estimation approach.
Spectral clustering is an alternative to the variational approach in Step 1.
\label{small_accuracy}}
\end{figure}
\begin{table}[h]
\centering
\begin{center}
\begin{tabular}{lccc}
\hline
                        & Bayesian approach & Variational approach & Spectral clustering\\
\hline
$n = 30$, $K = 3$, balanced      & 42,579.6 & 23.83      & .18 \\
\hline
$n = 30$, $K = 3$, unbalanced    & 32,480.9 & 24.00      & .18 \\
\hline
\end{tabular}
\s
\caption{\label{table:time}
Average computing time in seconds of the Bayesian approach,
the variational approach,
and spectral clustering.
$500$ graphs with $n=30$ nodes and $K=3$ blocks with 10 nodes (balanced case) or 5, 10, 15 nodes (unbalanced case) were generated.
The Bayesian approach is the gold standard.
The variational approach is the default in Step 1 of the two-step estimation approach.
Spectral clustering is an alternative to the variational approach in Step 1.
In Step 2 of the two-step estimation approach,
maximum pseudolikelihood estimates are computed.
The computing times of the variational approach and spectral clustering mentioned above are the total computing times,
including both Step 1 and Step 2.}
\end{center}
\end{table}

In Step 1 of the two-step estimation approach,
we use the variational approach described in Section \ref{sec:variational}.
We compare the variational approach to the spectral clustering method described in \citet{LeRi13}.
Spectral clustering is an alternative to the variational approach.
In Step 2 of the two-step estimation approach,
we use maximum pseudolikelihood estimators,
which are more scalable than Monte Carlo maximum likelihood estimators and facilitate simulation studies with up to 10,000 nodes.
In each scenario,
we generate 500 graphs from the model having between-block edge terms and within-block edge and transitive edge terms,
as described in Section \ref{sec:hergm}.
\alert{To select sensible values of the parameter vector $\btheta$,
note that the probabilistic behavior of random graphs with local dependence is sensitive to the choice of parameter values,
and so is the block recovery.
The same applies to stochastic block models: 
e.g.,
when the probabilities of edges within and between blocks are too low,
we may not be able to recover the block structure with high probability \citep{ZhZh16,Gaetal17}.
We have therefore selected the parameter vector $\btheta$ based on the following considerations:
We would like to ensure that,
with high probability, 
the model generates graphs that resemble real-world networks in terms of sufficient statistics $s(\bX)$.
In principle, 
we could select $\btheta$ by inspecting the expectation of $s(\bX)$.
The problem is that the expectation is not available in closed form.
We have therefore selected $\btheta$ based on simulating graphs,
and checking whether the simulated graphs resemble real-world networks.}

\begin{figure}[t]
\centering
\includegraphics[width = .48 \textwidth, height = .25 \textheight, keepaspectratio]{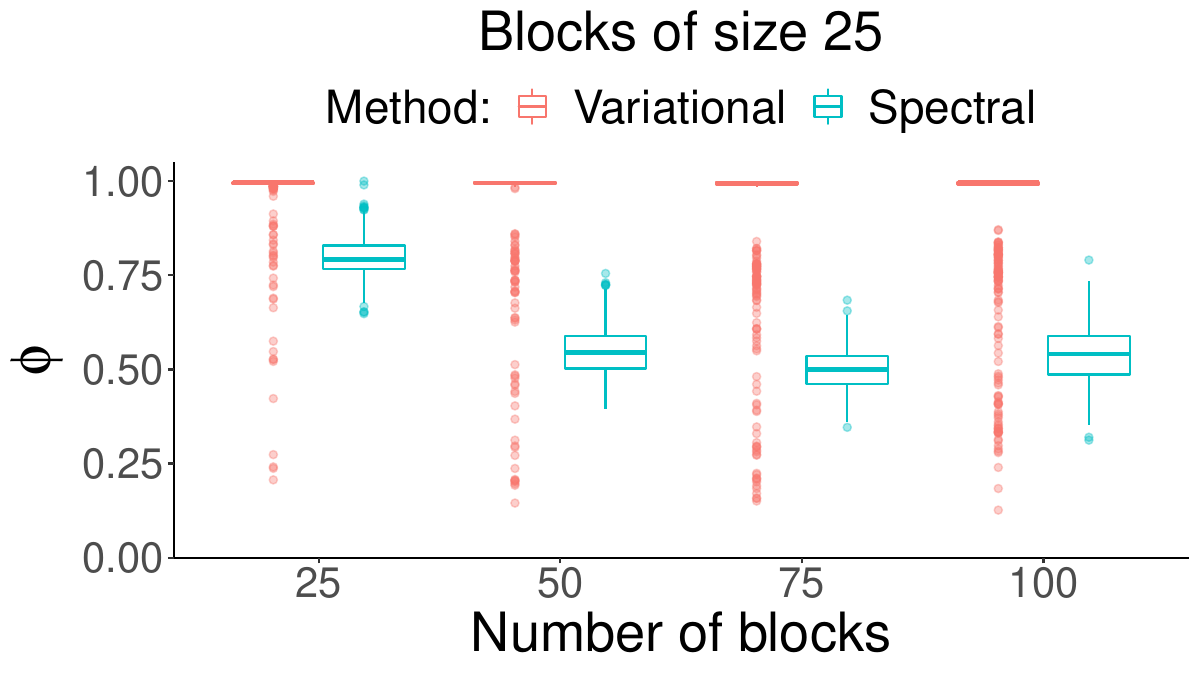} %
\includegraphics[width = .48 \textwidth, height = .25 \textheight, keepaspectratio]{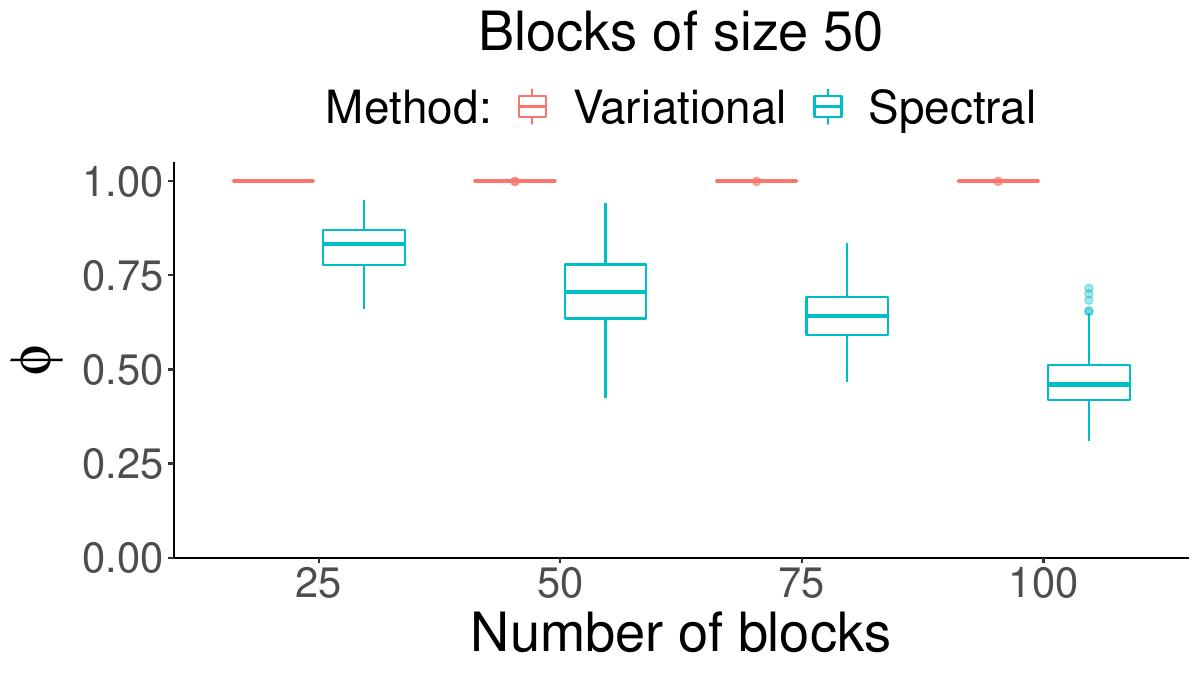} \s \\ 
\includegraphics[width = .48 \textwidth, height = .25 \textheight, keepaspectratio]{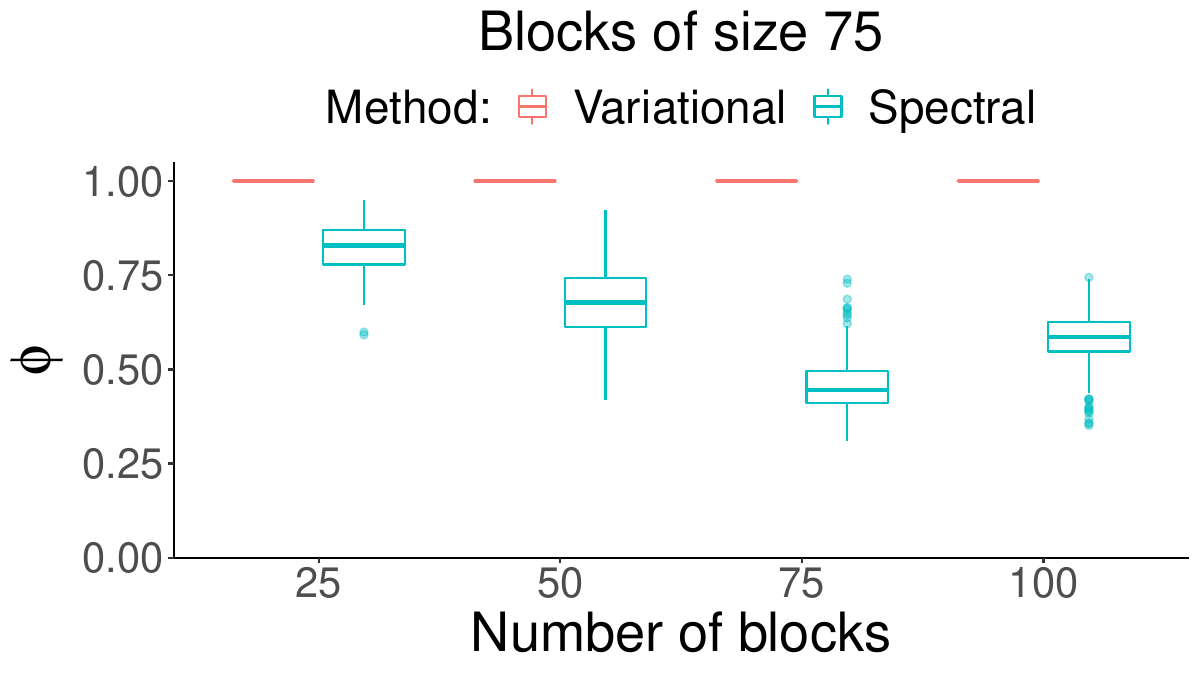} %
\includegraphics[width = .48 \textwidth, height = .25 \textheight, keepaspectratio]{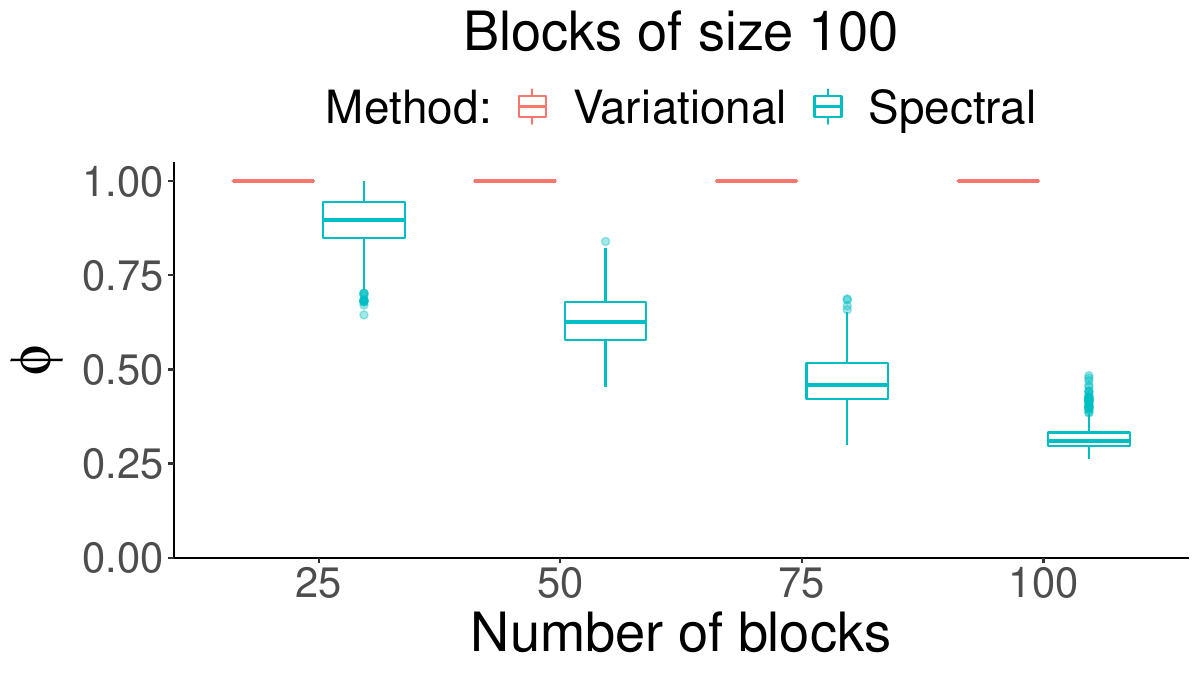}
\caption{Block recovery in terms of Yule's $\phi$-coefficient.
$500$ graphs with $K = 25, 50, 75, 100$ blocks of size $25, 50, 75, 100$ were generated from the model with between-block edge parameters $\theta_{B,k,l} = .5 - \log n$ ($l < k = 1, \dots, K$) and within-block edge and transitive edge parameters $\theta_{W,k,k,1} = -1.5$ and $\theta_{W,k,k,2} =  .5$ ($k = 1, \dots, K$).
The figure compares two alternative approaches to recovering the block structure in Step 1 of the two-step estimation approach,
the variational approach and spectral clustering.
\label{large_accuracy}
}
\end{figure}

In simulation study I,
to allow blocks of different sizes to have different parameters,
we use size-dependent between-block edge parameters $\theta_{B,k,l} = -.882 \log n$ ($l < k = 1, \dots, K$) and within-block edge  and transitive edge parameters $\theta_{W,k,k,1} = -.434\, \log n_k(\bz)$ and $\theta_{W,k,k,2} = .217 \, \log n_k(\bz)$ ($k = 1, \dots, K$),
where $n_k(\bz)$ is the size of block $k$ under $\bz \in \mbZ$.
The size-dependent parameterization is motivated by the sparsity of random graphs:
e.g.,
if edges $X_{i,j}$ are independent Bernoulli$(\mu)$ random variables,
it is tempting to believe that there exist constants $c > 0$ and $0 < \alpha < 1$ such that the expected number of edges of each node $(n - 1)\, \mu$ is bounded above by $c\; n^{\alpha}$,
because real-world networks are sparse.
As a consequence,
$\mu$ should be of order $n^{\theta}$ and $\eta = \logit(\mu)$ should be of order $\log n^{\theta} = \theta \log n$,
where $\theta = \alpha - 1 < 0$.
In more general models with edge terms as well as other model terms,
all model terms should scale as the edge term,
so that no model term can dominate any other model term.
These considerations suggest that parameters should scale with the log number of nodes.
In simulations studies II and IV,
the sizes of blocks are identical,
so we use between-block edge parameters $\theta_{B,k,l} = .5 - \log n$ ($l < k = 1, \dots, K$) and within-block edge and transitive edge parameters $\theta_{W,k,k,1} = -1.5$ and $\theta_{W,k,k,2} =  .5$ ($k = 1, \dots, K$).
In simulation study III,
we use the same within-block parameters,
but between-block edge parameters $\theta_{B,k,l} = .5 - \alpha\, \log n$ ($l < k = 1, \dots, K$),
with $\alpha$ varying from $.5$ and $1$.

\begin{figure}[t]
\centering
\includegraphics[width=.95\textwidth, keepaspectratio]{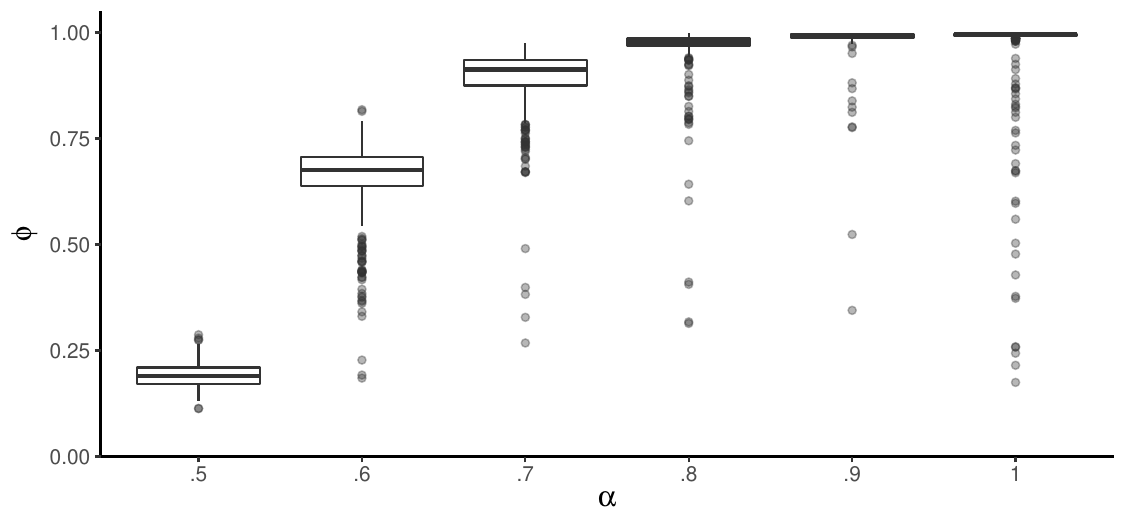}
\caption{Block recovery in terms of Yule's $\phi$-coefficient as a function of between-block subgraph sparsity.
$500$ graphs with $K = 25$ blocks of size $25$ were generated from the model with between-block edge parameters $\theta_{B,k,l} = .5 - \alpha \, \log n$ ($l < k = 1, \dots, K$) and within-block edge and transitive edge parameters $\theta_{W,k,k,1} = -1.5$ and $\theta_{W,k,k,2} =  .5$ ($k = 1, \dots, K$),
with $\alpha$ varying from $.5$ to $1$.
The figure shows that the more sparse between-block subgraphs are,
the more dense within-block subgraphs ``stand out,"
improving block recovery.
\label{betweenplot}
}
\end{figure}

\begin{figure}[t]
\centering
\includegraphics[width=.95\textwidth, keepaspectratio]{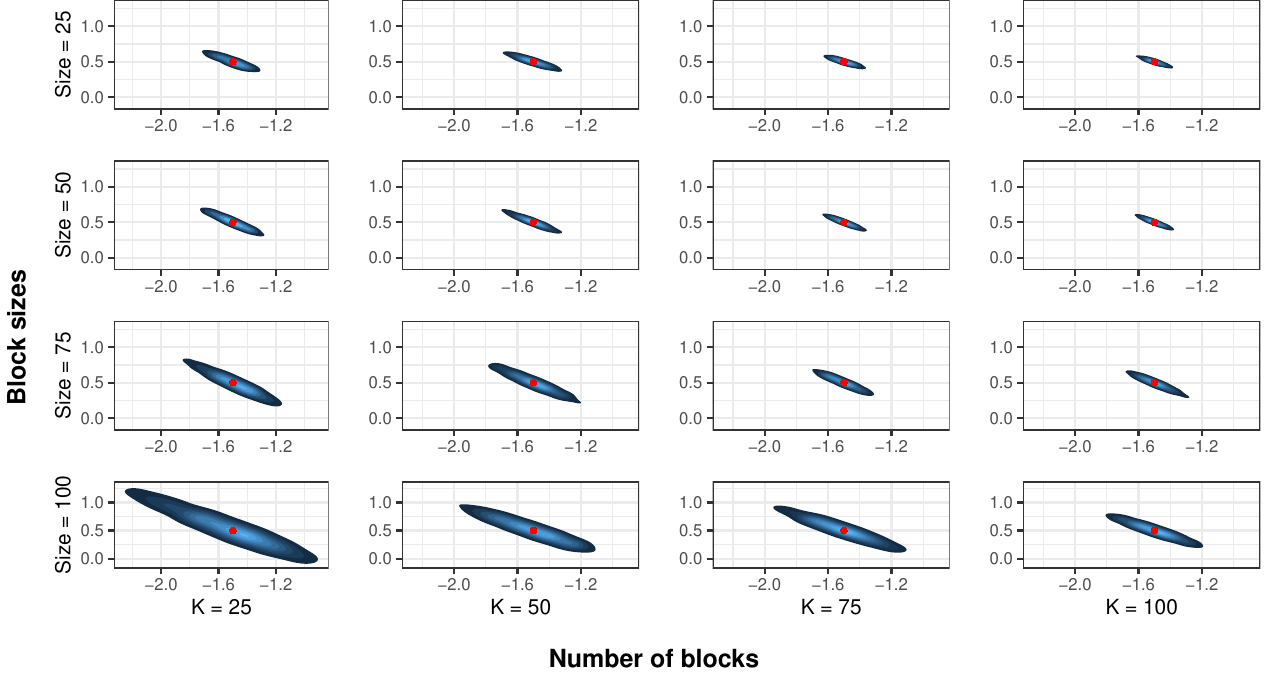}
\caption{Maximum pseudolikelihood estimates of within-block parameters.
$500$ graphs with $K = 25, 50, 75, 100$ blocks of size $25, 50, 75, 100$ were generated from the model with between-block edge parameters $\theta_{B,k,l} = .5 - \log n$ ($l < k = 1, \dots, K$) and within-block edge and transitive edge parameters $\theta_{W,k,k,1} = -1.5$ and $\theta_{W,k,k,2} =  .5$ ($k = 1, \dots, K$).
The horizontal axis corresponds to $\theta_{W,k,k,1} = -1.5$,
while the vertical axis corresponds to $\theta_{W,k,k,2} = .5$ ($k = 1, \dots, K$).
The red circle in each plot indicates the data-generating within-block parameter vector $(\theta_{W,k,k,1}, \theta_{W,k,k,2}) = (-1.5, .5)$.
\label{params}
}
\end{figure}

\hide{
\begin{figure}[t]
\centering
\includegraphics[width = .95 \textwidth, height = .3 \textheight, keepaspectratio]{alt_plot.pdf} %
\caption{Block recovery in terms of Yule's $\phi$-coefficient.
$500$ graphs with $K = 25, 50$ blocks of size $25$ were generated from the edge and transitive edge model with two different parameterizations of the within-block edge and transitive edge parameters. (Left) Results from the model with the between-block edge parameters $\theta_{B,k,l} = .5 - \log n$ ($l < k = 1, \dots, K$) and within-block edge parameters $\theta_{W,k,k,1} = -2$ and transitive edge parameters $\theta_{W,k,k,2} =  1$ ($k = 1, \dots, K$). (Right) Results from the model with the between-block edge parameters $\theta_{B,k,l} = .5 - \log n$ ($l < k = 1, \dots, K$) and within-block edge parameters $\theta_{W,k,k,1} = -1.5$ and transitive edge parameters $\theta_{W,k,k,2} = -.5$ ($k = 1, \dots, K$).
The figure compares how positive and negative parameterizations of transitivity affect the recovery of the block structure in Step 1 of the two-step estimation approach using the variational approach. 
\label{alt_param}
}
\end{figure}
}

In all scenarios,
we assess block recovery by using Yule's $\phi$-coefficient: 
\beno
\phi(\zs, \bz) 
&=& \dfrac{n_{0,0}\, n_{1,1} - n_{0,1}\, n_{1,0}}{\sqrt{(n_{0,0} + n_{0,1})\, (n_{1,0} + n_{1,1})\, (n_{0,0} + n_{1,0})\, (n_{0,1} + n_{1,1})}},
\ee
where
\beno
n_{a,b} 
&\equiv& n_{a,b}(\zs, \bz)
&=& \dsum_{i<j}^{n} \one(\one(\bz_i^\star = \bz_j^\star)=a)\; \one(\one(\bz_i=\bz_j)=b), 
& a, b \in \{0, 1\}.
\ee
Here,
$\one(.)$ is an indicator function,
which is $1$ if the statement in parentheses is true and is $0$ otherwise. 
\alert{The quantity $n_{0,0}(\zs, \bz)$ is the number of pairs of nodes that are assigned to distinct blocks under both $\zs$ and $\bz$,
while the quantity $n_{1,1}(\zs, \bz)$ is the number of pairs of nodes that are assigned to the same block under both $\zs$ and $\bz$.
The sum of the other two quantities, 
$n_{0,1}(\zs, \bz) + n_{1,0}(\zs, \bz)$,
 is the number of pairs of nodes on which $\zs$ and $\bz$ disagree.
If $\zs$ and $\bz$ agree on all pairs of nodes,
Yule's $\phi$-coefficient is $1$.
In fact,
Yule's $\phi$-coefficient is bounded above by $1$,
and is invariant to the labeling of the blocks.} 

The results of the small-scale simulation study I are shown in Figure~\ref{small_accuracy}.
The results suggest that the two-step estimation approach is almost as good as the Bayesian approach in terms of block recovery in the balanced case ($K = 3$ blocks of size $10$),
but is worse in the unbalanced case ($K = 3$ blocks of size $5, 10, 15$).
The worse performance in the unbalanced case may be due to the fact that there are smaller blocks in the unbalanced case than in the balanced case, 
and recovering small blocks is more challenging than recovering large blocks.
That said,
the advantage of the Bayesian approach over the two-step estimation approach in terms of block recovery in the unbalanced case is limited,
and comes at excessive costs:
Table \ref{table:time} reveals that the computing time of the Bayesian approach is thousands of times higher than the computing time of the two-step estimation approach.

The large-scale simulation studies II, III, and IV shed light on the performance of Steps 1 and 2 of the two-step estimation approach in large networks.
The results of simulation study II,
shown in Figure~\ref{large_accuracy},
reveal that Step 1 of the two-step estimation approach outperforms spectral clustering in terms of block recovery in most scenarios.
Simulation study III shows how the block recovery in Step 1 of the two-step estimation approach is affected by between-block subgraph sparsity.
According to Figure \ref{betweenplot},
the more sparse between-block subgraphs are,
the more dense within-block subgraphs ``stand out,"
improving block recovery.
Last, 
but not least, 
simulation study IV helps assess parameter recovery in Step 2 of the two-step estimation approach.
Figure \ref{params} shows that maximum pseudolikelihood estimates are close to the data-generating parameters,
and more so when the blocks are small and the number of blocks is large.

\section{Amazon product recommendation network}
\label{sec:application}

We use the two-step estimation approach to shed light on the structure of a large Amazon product recommendation network that is not captured by stochastic block models.
The data on the Amazon product recommendation network were collected by \citet{JaLe15} and can be downloaded from the website
\begin{center}
\url{http://snap.stanford.edu/data/com-Amazon.html}
\end{center}
The network consists of products listed at \url{www.amazon.com}.
Two products $i$ and $j$ are connected by an edge if $i$ and $j$ are frequently purchased together according to the ``Customers Who Bought This Item Also Bought" feature at \url{www.amazon.com}. 
Amazon assigns all products to categories,
which we consider to be ground-truth blocks. 
We use a subset of the network consisting of the top 500 non-overlapping categories with 10 to 80 products,
where the ranking of categories is based on \citet{JaLe15}. 
The resulting network consists of 10,448 products and 33,537 edges.

To model the Amazon product recommendation network,
we take advantage of curved exponential-family random graph models with within-block edge and geometrically weighted degree and edgewise shared partner terms \citep{SnPaRoHa04,HuHa04,HuGoHa08}.
The resulting models are more general than the motivating example used in Sections \ref{sec:hergm} and \ref{sec:approximation}---the model with between-block edge terms and within-block edge and transitive edge terms---and ensure that,
for each pair of products, 
the added value of additional edges and triangles within blocks decays.
In fact,
transitive edge terms are special cases of geometrically weighted edgewise shared partner terms,
and both of them are well-behaved alternatives to the ill-behaved triangle terms mentioned in Section \ref{sec:hergm}.
A full-fledged discussion of those models is beyond the scope of our paper.
We refer the interested reader to the seminal papers of \citet{SnPaRoHa04}, \citet{HuHa04}, and \citet{HuGoHa08}.

The natural parameters of the within-block edge terms are given by
\beno
\eta_{W,k,1}(\btheta, \bz) 
\= \theta_1\, \log n_k (\bz),
\ee
where the logarithmic term $\log n_k(\bz)$ arises from sparsity considerations and is a simple form of a size-dependent parameterization,
as explained in Section \ref{sec:sim_full}.
The within-block geometrically weighted degree terms are based on the number of products with $t$ edges in block $\mA_k$.
The natural parameters of within-block geometrically weighted degree terms are given by
\beno
\eta_{W,k,2,t}(\btheta, \bz) 
&=&\theta_{2}\, \log n_k (\bz)\, \exp(\theta_3)\, \left[1 - \left(1 - \exp(-\theta_3)\right)^t\right],
& t = 1, \dots, n_k(\bz) - 1.
\ee
The within-block geometrically weighted edgewise shared partner terms are based on the number of connected pairs of products $i$ and $j$ in block $\mA_k$ such that $i$ and $j$ have $t$ shared partners in block $\mA_k$.
The natural parameters of the within-block geometrically weighted edgewise shared partner terms are given by
\beno
\eta_{W,k,3,t}(\btheta, \bz) 
&=&\theta_{4}\, \log n_k (\bz)\, \exp(\theta_5)\, \left[1 - \left(1 - \exp(-\theta_5)\right)^t\right],
& t = 1, \dots, n_k(\bz) - 2.
\ee
To reduce computing time, 
it is convenient to truncate the two geometrically weighted model terms by setting $\eta_{W,k,2,t}(\btheta, \bz) = 0$ ($t = 21, \dots, n_k(\bz) - 1$) and $\eta_{W,k,3,t}(\btheta, \bz) = 0$ ($t = 13, \dots, n_k(\bz) - 2$).
The two thresholds $21$ and $13$ are motivated by the fact that no product has $21$ or more edges and less than 1\% of all pairs of products has $13$ or more edgewise shared partners.
Last,
but not least,
the natural parameters of the between-block edge terms are given by
\beno
\eta_{B,k,l}(\btheta, \bz) 
\= \theta_6\, \log n,
& l < k.
\ee
The resulting exponential family is a curved exponential family \citep{HuHa04},
because the natural parameter vector $\bta(\btheta, \bz)$ of the exponential family is a nonlinear function of $\btheta$ given $\bz \in \mbZ$.
The natural parameter vector $\bta(\btheta, \bz)$ depends on the sizes of blocks,
because we do not want to force small and large blocks to have the same natural parameters,
as explained in Section \ref{sec:sim_full}.

The curved exponential-family model specified above can capture an excess in the expected number of triangles within blocks relative to stochastic block models,
while ensuring that,
for each pair of products, 
the added value of additional edges and triangles within blocks decays \citep{StScBoMo18}.
An excess in the expected number of triangles within blocks relative to stochastic block models can arise when,
e.g.,
(a) three products are similar (e.g., three books on the same topic);
(b) three products are dissimilar but complement each other (e.g., a bicycle helmet, head light, and tail light);
and
(c) three products, 
either similar or dissimilar, 
were produced by the same source (e.g., three books written by the same author).

\begin{table}[t]
\centering
 \begin{tabular}{lrrrr}\hline
             Term & Estimate & S.E. & Estimate & S.E.  \\ \hline 
                         Within-block edges $\theta_1$           				     & $-.369$      &  $.002$  &-1.410      & $.008$	\\ \hline
                         Within-block degrees $\theta_2$ (base parameter)     	             & 	 	&              & .742	    & $.020$	\\ 
			 Within-block degrees $\theta_3$ (decay parameter)            	     &	 		& 		& .910	    &$.023$	\\ \hline
                         Within-block shared partners $\theta_4$ (base parameter) 	     & 	 	& 		& .303	   &$.005$	\\
			 Within-block shared partners $\theta_5$ (decay parameter)		     &	 		& 		&  1.106   	   & 	$.012$\\  \hline
 			 Between-block edges $\theta_6$ 			     &$-1.199$	& $.004$    &  $-1.199$  	   & 	$.004$ \\  \hline

 \end{tabular}
 \s
\caption{\label{table:params}
Monte Carlo maximum likelihood estimates and standard errors (S.E.) of $\theta_1, \dots, \theta_6$ estimated from the Amazon product recommendation network with 10,448 products; 
note that $\btheta = (\theta_1, \dots, \theta_6)$ should not be confused with the size-dependent natural parameter vector $\bta(\btheta, \bz)$. 
\alert{The parameters $\theta_1$ and $\theta_6$ are the base weights of the within- and between-block edge terms,
respectively.
The parameters $\theta_2$ and $\theta_4$ are the base weights of the within-block degree and edgewise shared partner terms,
whereas $\theta_3$ and $\theta_5$ control the rate of decay of the added value of additional edges and edgewise shared partners, 
respectively.}}
\end{table}

\begin{figure}[t]
\centering
    \includegraphics[scale=.41]{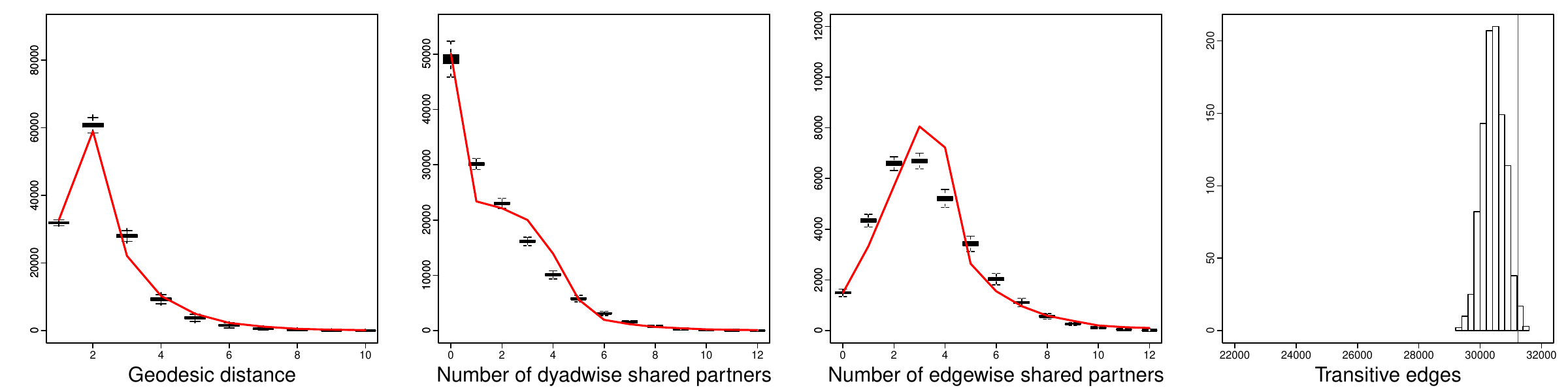}
\caption{Amazon product recommendation network with 10,448 products: goodness-of-fit of curved exponential-family random graph model.
The red lines indicate observed values of statistics.
\label{gof:ergm}
}

\includegraphics[scale=.41]{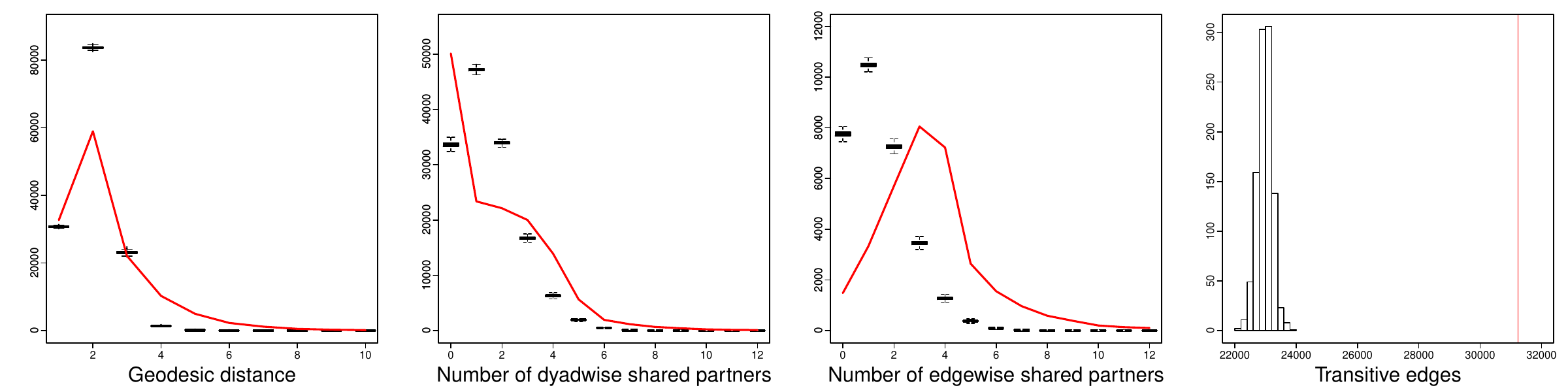}
\caption{Amazon product recommendation network with 10,448 products: goodness-of-fit of stochastic block models. 
The red lines indicate observed values of statistics.
\label{gof:sbm}
}
\end{figure}

Since we know the number of ground-truth blocks,
we set $K=500$ and estimate the block structure by using the two-step estimation approach,
using the variational approach in Step 1 and Monte Carlo maximum likelihood estimates in Step 2.
To assess the performance of the two-step estimation approach in terms of block recovery,
we use Yule's $\phi$-coefficient.
Yule's $\phi$-coefficient turns out to be $.964$,
which indicates near-perfect recovery of the ground-truth block structure.
\alert{An inspection of the estimated block structure reveals that,
out of the 500 categories of products in the Amazon product recommendation network, 
products in 15 categories are misclassified.
Some of the small categories are merged with large categories,
while some unusual products of large blocks are misclassified as well.
Some of the products are unusual in the sense of having few edges to other products in the same category,
while others are unusual in the sense of having many edges to other products in the same category.}

The Monte Carlo maximum likelihood estimates and standard errors of $\theta_1, \dots, \theta_6$ are shown in Table \ref{table:params}. 
\alert{The parameters $\theta_1, \dots, \theta_6$ of the geometrically weighted terms can be interpreted in terms of log odds of conditional probabilities of edges,
given all other edges \citep{Hu08}.
A work-out example can be found in \citet{StScBoMo18}.}
Table \ref{table:params} suggests that there is evidence for transitivity.
The observed tendency towards transitivity has advantages in practice: 
\alert{It enables Amazon to recommend co-purchases of brandnew products and existing products,
even when there are no data on past co-purchases of these products.
For example,
when a brandnew product $i$ is introduced (e.g., a novel) and product $i$ is known to be related to existing product $j$ (e.g., a novel by the same author),
and product $j$ tends to be co-purchased with product $k$ (e.g., a classic novel),
then Amazon could recommend co-purchases of products $i$ and $j$ as well as products $i$ and $k$,
despite the fact that there are no data on past co-purchases of products $i$ and $k$ and there is no direct connection between products $i$ and $k$ (although there in indirect connection via product $j$).} 

To demonstrate that the curved exponential-family random graph model considered here can capture structural features of networks that simpler models---such as stochastic block models---cannot capture,
we compare the goodness-of-fit of the curved exponential-family random graph model to the goodness-of-fit of stochastic block models.
Since the two models impose the same probability law on between-block subgraphs,
it is natural to compare the two models in terms of goodness-of-fit with respect to within-block subgraphs.
We assess the goodness-of-fit of the two models in terms of the within-block geodesic distances of pairs of products,
i.e., 
the length of the shortest path between pairs of products in the same block;
the numbers of within-block dyadwise shared partners,
i.e.,
the number of unconnected or connected pairs of products with $i$ shared partners in the same block;
 the numbers of within-block edgewise shared partners,
i.e.,
the number of connected pairs of products with $i$ shared partners in the same block;
and the number of within-block transitive edges,
i.e., 
the number of pairs of products with at least one shared partner in the same block.
Figures \ref{gof:ergm} and \ref{gof:sbm} compare the goodness-of-fit of the two models based on 1,000 graphs simulated from the estimated models.
The figures suggest that 
the curved exponential-family random graph model considered here is superior to the stochastic block model in terms of both connectivity and transitivity.

\section{Discussion}
\label{sec:discussion}

The two-step estimation approach proposed here enables large-scale estimation of models with local dependence and unknown block structure provided that the number of blocks $K$ is known.
An important direction of future research are methods for selecting $K$ when $K$ is unknown.
We note that even in the special case of stochastic block models,
the issue of selecting $K$ has not received much attention,
with the notable exception of recent work by \citet{Sa15}, \citet{Wa15}, and others.
Whether---and how---the developed methods can be extended to models with local dependence is an open question,
but having scalable methods for selecting $K$ would doubtless be useful in practice.

\alert{The proposed methods are implemented in {\tt R} packages {\tt hergm} \citep{ScLu15}.
A stable and user-friendly version will be released in the near future.}

\subsection*{Acknowledgements}

We acknowledge support from the National Science Foundation (NSF) in the form of NSF awards DMS-1513644 (SB, JS, MS) and DMS-1812119 (JS, MS).

\newpage

\setcounter{page}{1}

\begin{center}
\LARGE
{Supplement:

\longtitle\s\s
\\
\normalsize
{\large 
Sergii Babkin\hspace{2cm}Jonathan R.\ Stewart\hspace{2cm}Xiaochen Long\hspace{2cm}Michael Schweinberger
}
}
\end{center}

\section*{Proofs of theoretical results}

We prove Theorems \ref{t.concentration2} and \ref{t.concentration} of Section \ref{sec:approximation}.
To prove them,
we need three additional results,
Lemma \ref{l.concentration} and Propositions \ref{p.concentration} and \ref{p.expectation}.
To state them,
let
\beno
g(\bx; \btheta, \bz) 
&=& \log p_{\bta(\btheta,\bz)}(\bx) - \log p_{\bta(\btheta_1,\btheta_2=\bm{0},\bz)}(\bx),
\ee
where $g(\bx; \btheta, \bz)$ is considered as a function of $\bx \in \mbX$ for fixed $(\btheta, \bz) \in \bTheta \times \mbZ$.
Observe that the expectation $\mbE\, s(\bX)$ of the sufficient statistic vector $s(\bX)$ exists as long as the data-generating natural parameter vector $\bta(\bthetas, \bz^\star)$ is contained in the interior of the natural parameter space of the exponential family \citep[][Theorem 2.2, pp.\ 34--35]{Br86}.
As a consequence,
the expectations $\mbE \log p_{\bta(\btheta, \bz)}(\bX)$ and $\mbE\, g(\bX; \btheta, \bz)$ exist,
because
\beno
\mbE \log p_{\bta(\btheta, \bz)}(\bX)
\= \langle\bta(\btheta, \bz),\, \mbE\, s(\bX)\rangle - \psi(\bta(\btheta, \bz))
\ee
and
\beno
\mbE\, g(\bX; \btheta, \bz)
\hide{
&=& \mbE\, \left[\mbE \log p_{\bta(\btheta, \bz)}(\bX) - \log p_{\bta(\btheta_1,\btheta_2=\bm{0},\bz)}(\bX)\right]
}
&=& \mbE \log p_{\bta(\btheta, \bz)}(\bX) - \mbE \log p_{\bta(\btheta_1,\btheta_2=\bm{0},\bz)}(\bX).
\ee

We first state Lemma \ref{l.concentration} and Propositions \ref{p.concentration} and \ref{p.expectation} and then prove Theorems \ref{t.concentration2} and \ref{t.concentration}.

\begin{lemma}
\label{l.concentration}
Consider a model with local dependence.
Let $f: \mX \times \mbZ \mapsto \mbR$ be a function of within-block edge variables $(X_{i,j})_{i < j:\, \bz_i = \bz_j}^n$ that is Lipschitz with respect to the Hamming metric $d: \mX \times \mX \mapsto \{0, 1, 2, \dots\}$ with Lipschitz coefficient $\norm{f}_{\lip} > 0$ and $\mbE\, f(\bX; \bz) < \infty$.
Then there exists a universal constant $c > 0$ such that,
for all $\bz \in \mbZ$ and all $t > 0$,
\beno
\mbP \left(|f(\bX; \bz) - \mbE\, f(\bX; \bz)|\; \geq\; t\right)
&\lte& 2\, \exp\left(- \dfrac{t^2}{c\, K\, m(\bz)^2\, \norm{\mA}_\infty^4 \, \norm{f}^2_{\lip}}\right).
\ee
\end{lemma}

\llproof \ref{l.concentration}.
The proof of Lemma \ref{l.concentration} resembles the proof of Proposition 1 of \citet{ScSt16} step by step and is therefore omitted.

\s
 
\begin{proposition}
\label{p.concentration}
Consider a model with local dependence satisfying conditions [C.1] and [C.2].
Choose $\epsilon \in (0, 1)$ as small as desired.
Then there exists a universal constant $c>0$ such that
\beno\mbP\left(\max\limits_{\bz \in \mbS} |g(\bX; \btheta, \bz) - \mbE\, g(\bX; \btheta, \bz)|\; \geq\; c\; \sqrt{-\log\dfrac{\epsilon}{2} + n \log K}\; \sqrt{K}\; m(\mbS)^2\; \norm{\mA}_\infty^2\; \log n\right)
\;\leq\; \epsilon.
\ee
\end{proposition}

\pproof \ref{p.concentration}.
To show that the probability mass of $g(\bX; \btheta, \bz)$ concentrates around its expectation $\mbE\, g(\bX; \btheta, \bz)$,
observe that the Lipschitz coefficient of the function $g: \mX \times \bTheta \times \mbZ \mapsto \mbR$ with respect to the Hamming metric $d: \mX \times \mX \mapsto \{0, 1, 2, \dots\}$ is defined by
\beno
\norm{g}_{Lip} 
&=& \underset{(\bx_1,\bx_2)\, \in\, \mbX \times \mbX:\; d(\bx_1,\bx_2)>0}{\sup} \dfrac{|g(\bx_1; \btheta, \bz)-g(\bx_2; \btheta, \bz)|}{d(\bx_1,\bx_2)}.
\ee
Since the term $\psi(\bta(\btheta, \bz)) - \psi(\bta(\btheta_1, \btheta_2\hspace{-.1cm}=\hspace{-.1cm}\bm{0}, \bz))$ of $g(\bx_1; \btheta, \bz)$ and $g(\bx_2; \btheta, \bz)$ cancels,
we obtain
\beno
\dfrac{|g(\bx_1; \btheta, \bz)-g(\bx_2; \btheta, \bz)|}{d(\bx_1,\bx_2)} 
\hide{
&=& \dfrac{| \left\langle\bta(\btheta,\bz) - \bta(\btheta_1, \btheta_2\hspace{-.1cm}=\hspace{-.1cm}\bm{0}, \bz), s(\bx_1)\right\rangle - \left\langle\bta(\btheta, \bz) - \bta_0(\btheta,\bz), s(\bx_2)\right\rangle|}{d(\bx_1,\bx_2)}\s
\\
}
&=& \dfrac{|\langle\bta(\btheta,\bz) - \bta(\btheta_1, \btheta_2\hspace{-.1cm}=\hspace{-.1cm}\bm{0}, \bz), s(\bx_1) - s(\bx_2)\rangle|}{d(\bx_1,\bx_2)}.
\ee
By condition [C.1] and the fact that $\bta(\btheta,\bz) - \bta(\btheta_1, \btheta_2\hspace{-.1cm}=\hspace{-.1cm}\bm{0}, \bz) \in \mR^{\qqq}$,
there exists $c_0 > 0$ such that
\beno
\dfrac{|g(\bx_1; \btheta, \bz) - g(\bx_2; \btheta, \bz)|}{d(\bx_1,\bx_2)} 
= \dfrac{|\langle\bta(\btheta,\bz) - \bta(\btheta_1, \btheta_2\hspace{-.1cm}=\hspace{-.1cm}\bm{0}, \bz), s(\bx_1) - s(\bx_2)\rangle|}{d(\bx_1,\bx_2)}
\leq c_0\, m(\bz) \log n.
\ee
Thus,
the Lipschitz coefficient of the function $g: \mbX \times \bTheta \times \mbZ$ is bounded above by
\beno
\norm{g}_{Lip}
&\leq& c_0\; m(\bz)\, \log n 
&\leq& c_0\; m(\mbS)\, \log n 
& \mbox{for all}
& \bz \in \mbS.
\ee
By applying Lemma \ref{l.concentration} to the function $g: \mbX \times \bTheta \times \mbZ$ of within-block edges with Lipschitz coefficient $\norm{g}_{Lip} \leq c_0\; m(\mbS)\, \log n$,
we can conclude that there exists $c > 0$ such that,
for all $t > 0$,
\beno
\mbP\left(|g(\bX; \btheta, \bz) - \mbE\, g(\bX; \btheta, \bz)|\; \geq\; t\right)
\lte 
\hide{
2\, \exp\left(- \dfrac{t^2}{c^2\, K\, m(\mbS)^2\, \norm{\mA}_\infty^4\, m(\mbS)^2}\right)\s
\\
\= 
}
2\, \exp\left(- \dfrac{t^2}{c^2\, K\; m(\mbS)^4\; \norm{\mA}_\infty^4\; (\log n)^2}\right).
\ee
A union bound over the $|\mbS| \leq K^n$ block structures shows that
\beno
\mbP\left(\max\limits_{\bz \in \mbS} |g(\bX; \btheta, \bz) - \mbE\, g(\bX; \btheta, \bz)| \geq t\right)
\lte 2\, \exp\left(- \dfrac{t^2}{c^2\, K\; m(\mbS)^4\; \norm{\mA}_\infty^4\; (\log n)^2} + n \log K\right).
\ee
Choose $\epsilon \in (0, 1)$ as small as desired and let 
\beno
t
\= c\; \sqrt{-\log\dfrac{\epsilon}{2} + n \log K}\; \sqrt{K}\; m(\mbS)^2\; \norm{\mA}_\infty^2\; \log n
&>& 0,
\ee
where $c > 0$.
Then 
\beno
\mbP\left(\max\limits_{\bz \in \mbS} |g(\bX; \btheta, \bz) - \mbE\, g(\bX; \btheta, \bz)|\; \geq\; c\; \sqrt{-\log\dfrac{\epsilon}{2} + n \log K}\; \sqrt{K}\; m(\mbS)^2\; \norm{\mA}_\infty^2\; \log n\right) 
\;\leq\; \epsilon.
\ee

\hide{

\newpage

Remark:
We want
\beno
2\, \exp\left(- \dfrac{t^2}{c^2\, K\; m(\mbS)^4\; \norm{\mA}_\infty^4\; (\log n)^2} + n \log K\right)
\lte \epsilon
\ee
or
\beno
- \dfrac{t^2}{c^2\, K\; m(\mbS)^4\; \norm{\mA}_\infty^4\; (\log n)^2} + n \log K
\lte \log\dfrac{\epsilon}{2}
\ee
or
\beno
\dfrac{t^2}{c^2\, K\; m(\mbS)^4\; \norm{\mA}_\infty^4\; (\log n)^2} - n \log K
\gte -\log\dfrac{\epsilon}{2}
\ee
or
\beno
\dfrac{t^2}{c^2\, K\; m(\mbS)^4\; \norm{\mA}_\infty^4\; (\log n)^2}
\gte -\log\dfrac{\epsilon}{2} + n \log K
\ee
or
\beno
t^2
\gte \left(-\log\dfrac{\epsilon}{2} + n \log K\right)\; c^2\, K\; m(\mbS)^4\; \norm{\mA}_\infty^4\; (\log n)^2
\ee
or
\beno
t
\gte c\; \sqrt{-\log\dfrac{\epsilon}{2} + n \log K}\; \sqrt{K}\; m(\mbS)^2\; \norm{\mA}_\infty^2\; \log n.
\ee

\newpage

}

\s

\begin{proposition}
\label{p.expectation}
Consider a model with local dependence satisfying conditions [C.1] and [C.2].
Then there exists a universal constant $c > 0$ such that
\beno
\max\limits_{\bz \in \mbS} \left|\mbE\, g(\bX; \btheta, \bz)\right|
\lte c\; K\; m(\mbS)^2\; \log n.
\ee
\end{proposition}

\pproof \ref{p.expectation}.
By definition,
\beno
\mbE\, g(\bX; \btheta, \bz)
\= \mbE \log p_{\bta(\btheta, \bz)}(\bX) - \mbE \log p_{\bta(\btheta_1,\btheta_2=\bm{0},\bz)}(\bX),
\ee
where,
for all $(\btheta, \bz) \in \bTheta \times \mZ$,
\beno
p_{\bta(\btheta, \bz)} (\bx)
&=& \dprod_{k=1}^K p_{\bta_{W,k}(\btheta, \bz)}(\bx_{k,k})\; \dprod_{l=1}^{k-1}\;\, \dprod_{i, j:\; z_{i,k} = 1,\; z_{j,l} = 1}^n\, p_{\bta(\btheta, \bz)}(x_{i,j}),
&& \bx \in \mbX.
\ee
By construction of $p_{\bta(\btheta, \bz)}(\bx)$ and $p_{\bta(\btheta_1, \btheta_2=\bm{0}, \bz)}(\bx)$,
the contributions of between-block subgraphs to the loglikelihood function are the same under both models,
hence the expectation of the loglikelihood ratio reduces to the expectation of the loglikelihood ratio of within-block subgraphs:
\beno
\mbE \log p_{\bta(\btheta, \bz)}(\bX) - \mbE \log p_{\bta(\btheta_1,\btheta_2=\bm{0},\bz)}(\bX)
= \dsum_{k=1}^K \left[\mbE \log p_{\bta(\btheta,\bz)}(\bX_{k,k}) - \mbE \log p_{\bta(\btheta_1,\btheta_2=\bm{0},\bz)}(\bX_{k,k})\right].
\ee
By the triangle inequality,
\beno
\label{g_w}
\left|\mbE \log p_{\bta(\btheta, \bz)}(\bX) - \mbE \log p_{\bta(\btheta_1,\btheta_2=\bm{0},\bz)}(\bX)\right|
\leq \dsum_{k=1}^K \left|\mbE \log p_{\bta(\btheta,\bz)}(\bX_{k,k}) - \mbE \log p_{\bta(\btheta_1,\btheta_2=\bm{0},\bz)}(\bX_{k,k})\right|.
\ee
The terms $|\mbE \log p_{\bta(\btheta,\bz)}(\bX_{k,k}) - \mbE \log p_{\bta(\btheta_1,\btheta_2=\bm{0},\bz)}(\bX_{k,k})|$ can be bounded above by invoking the triangle inequality:
\beno
\left|\mbE \log p_{\bta(\btheta,\bz)}(\bX_{k,k}) - \mbE \log p_{\bta(\btheta_1,\btheta_2=\bm{0},\bz)}(\bX_{k,k})\right|
\lte |\langle\bta_{k,k}(\btheta_{k,k}, \bz) - \bta_{k,k}(\btheta_{k,k,0}, \bz),\, \mbE\, s_{k,k}(\bX)\rangle|\s
\\
&+& |\psi_{k,k}(\bta_{k,k}(\btheta_{k,k}, \bz)) - \psi_{k,k}(\bta_{k,k}(\btheta_{k,k,0}, \bz))|,
\ee
where $\bta_{k,k}(\btheta_{k,k}, \bz)$, $s_{k,k}(\bx)$, and $\psi_{k,k}(\bta_{k,k}(\btheta_{k,k}, \bz))$ are the natural parameter vector, the sufficient statistics vector, and the log-normalizing constant of $p_{\bta(\btheta,\bz)}(\bX_{k,k})$,
and $\btheta_{k,k,0} = (\theta_{W,k,1}, \btheta_{W,k,2}=\bm{0})$. 
We bound the two terms on the right-hand side of the inequality above one by one.

\s

{\em First term.} 
By condition [C.2],
there exist $c_1 > 0$ and $c_2 > 0$ such that
\beno
|\langle\bta_{k,k}(\btheta_{k,k}, \bz) - \bta_{k,k}(\btheta_{k,k,0}, \bz),\, \mbE\, s_{k,k}(\bX)\rangle|
\lte c_1\; \norm{\btheta_{k,k} - \btheta_{k,k,0}}\; m(\bz)^2\, \log n\s
\\
\lte c_2\; m(\mbS)^2\; \log n,
\ee
where the last inequality follows from the assumption that $\bTheta_{k,k}$ is compact.

\s

{\em Second term.} 
By the mean-value theorem along with classic exponential-family properties,
there exists $\dot\bta_{k,k} = \alpha\, \bta_{k,k}(\btheta_{k,k}, \bz) + (1 - \alpha)\, \bta_{k,k}(\btheta_{k,k,0}, \bz)$ ($0 < \alpha < 1$) such that
\beno
|\psi_{k,k}(\bta_{k,k}(\btheta_{k,k}, \bz)) - \psi_{k,k}(\bta_{k,k}(\btheta_{k,k,0}, \bz))|
\hide{
\\
\= |\langle\bta_{k,k}(\btheta_{k,k}, \bz) - \bta_{k,k}(\btheta_{k,k,0}, \bz),\; \nabla_{\bta_{k,k}}\, \psi_{k,k}(\bta_{k,k}(\btheta_{k,k}, \bz))|_{\bta_{k,k} = \dot\bta_{k,k}}\rangle|\s
\\
}
\;=\; |\langle\bta_{k,k}(\btheta_{k,k}, \bz) - \bta_{k,k}(\btheta_{k,k,0}, \bz),\; \mbE_{\dot\bta_{k,k}}\, s_{k,k}(\bX)\rangle|.
\ee
It is worth noting that the existence of $\dot\bta_{k,k}$ is ensured by the fact that the natural parameter space of exponential families is convex \citep[Theorem 1.13,][p.\ 19]{Br86},
which implies that an element $\dot\bta_{k,k}$ in the interior of the natural parameter exists as long as $\bta_{k,k}(\btheta_{k,k}, \bz)$ and $\bta_{k,k}(\btheta_{k,k,0}, \bz)$ are contained in the interior of the natural parameter space.
Therefore,
the second term can be bounded along the same lines as the first term,
which implies that there exist $c_3 > 0$ such that
\beno
|\psi_{k,k}(\bta_{k,k}(\btheta_{k,k}, \bz)) - \psi_{k,k}(\bta_{k,k}(\btheta_{k,k,0}, \bz))|
\lte c_3\; m(\mbS)^2\, \log n.
\ee

{\em Conclusion.}
Collecting terms shows that that there exist $c > 0$ such that
\hide{
\beno
|\mbE\; g(\bX; \btheta, \bz)|
\hide{
&\leq& \dsum_{k=1}^K |\langle\bta_{k,k}(\btheta_{k,k}, \bz) - \bta_{k,k}(\btheta_{k,k}, \bz),\, \mbE\, s_{k,k}(\bX)\rangle|\s
\\
&+& \dsum_{k=1}^K |\psi_{k,k}(\bta_{k,k}(\btheta_{k,k}, \bz)) - \psi_{k,k}(\bta_{k,k}(\btheta_{k,k}, \bz))|\s
\\
}
\lte c\; K\; m(\mbS)^2\; \log n
\ee
and
}
\beno
\max\limits_{\bz \in \mbS} |\mbE\; g(\bX; \btheta, \bz)|
\lte c\, K\, m(\mbS)^2\, \log n.
\ee

\s

Armed with Propositions \ref{p.concentration} and \ref{p.expectation},
we can prove Theorem \ref{t.concentration2} and \ref{t.concentration}.

\s

\ttproof \ref{t.concentration2}.
Theorem \ref{t.concentration2} is an application of Theorem \ref{t.concentration}.
To prove Theorem \ref{t.concentration2},
all we need to do verify conditions [C.1] and [C.2] of Theorem \ref{t.concentration}.

\s

{\em Condition [C.1].}
By the triangle inequality,
\beno
&& \left|\left\langle\bta(\btheta, \bz),\, s(\bx_1) - s(\bx_2)\right\rangle\right|
\;=\; \left|\dsum_{k \leq l}^K \left\langle\bta_{k,l}(\btheta),\, s_{k,l}(\bx_1) - s_{k,l}(\bx_2)\right\rangle\right|\s
\\
\lte \dsum_{k \leq l}^K \left|\left\langle\bta_{k,l}(\btheta),\, s_{k,l}(\bx_1) - s_{k,l}(\bx_2)\right\rangle\right|
\;\leq\; \dsum_{l < k}^K \left|\theta_{B,k,l}\right|\, |s_{B,k,l}(\bx_1) - s_{B,k,l}(\bx_2)|\s
\\
&+& \dsum_{k = 1}^K \left|\theta_{W,k,1} \right|\, |s_{W,k,1}(\bx_1) - s_{W,k,1}(\bx_2)|
\;+\; \dsum_{k = 1}^K \left|\theta_{W,k,2} \right|\, |s_{W,k,2}(\bx_1) - s_{W,k,2}(\bx_2)|.
\hide{
\\
\lte c_1\, d(\bx_1, \bx_2) + 2\, c_2\, d(\bx_1, \bx_2)\; m(\bz)\s
\\
&\leq& \max(c_1, 2\, c_2)\; d(\bx_1, \bx_2)\; m(\bz)
}
\ee
The first term is bounded above by $c_1\, d(\bx_1, \bx_2)$,
whereas the other two terms are bounded above by $c_2\, d(\bx_1, \bx_2)\; m(\bz)$,
which implies that
\beno
\left|\left\langle\bta(\btheta, \bz),\, s(\bx_1) - s(\bx_2)\right\rangle\right|
\lte c_1\, d(\bx_1, \bx_2) + 2\, c_2\, d(\bx_1, \bx_2)\; m(\bz)\s
\\
&\leq& \max(c_1,\, 2\, c_2)\; d(\bx_1, \bx_2)\; m(\bz)
\ee
by the compactness of $\bTheta_{k,l}$ ($k \leq l = 1, \dots, K$).
Thus,
condition [C.1] is satisfied.

\s

{\em Condition [C.2].}
Condition [C.2] is satisfied,
because
\beno
&& \left|\left\langle\bta_{k,l}(\btheta_{k,l,1}, \bz) - \bta_{k,l}(\btheta_{k,l,2}, \bz),\, \mbE_{\bta(\btheta, \bz)}\, s_{k,l}(\bX)\right\rangle\right|\s
\\
\lte |\theta_{1,k,l,1} - \theta_{1,k,l,2}|\, \mbE_{\bta(\btheta, \bz)}\, s_{B,k,l}(\bX) + |\theta_{2,k,l,1} - \theta_{2,k,l,2}|\; \mbE_{\bta(\btheta, \bz)}\, s_{2,k,l}(\bX)\s
\\
\lte |\theta_{1,k,l,1} - \theta_{1,k,l,2}|\; m(\bz)^2 + |\theta_{2,k,l,1} - \theta_{2,k,l,2}|\; m(\bz)^2\s
\\
\= \norm{\btheta_{k,l,1} - \btheta_{k,l,2}}_1\; m(\bz)^2
\;\leq\; \sqrt{2}\; \norm{\btheta_{k,l,1} - \btheta_{k,l,2}}\; m(\bz)^2,
\ee
where $\norm{.}_1$ denotes the $\ell_1$-norm whereas $\norm{.} \equiv \norm{.}_2$ denotes the $\ell_2$-norm of the vector of interest.

\s

\ttproof \ref{t.concentration}.
\hide{
Observe that
\beno
|g(\bX; \btheta, \bz)|
\lte |g(\bX; \btheta, \bz) - \mbE\, g(\bX; \btheta, \bz)| + |\mbE\, g(\bX; \btheta, \bz)|.
\ee
}
Observe that,
for all $t > 0$,
\beno
&& \mbP\left(\max\limits_{\bz\in\mbS} |g(\bX; \btheta, \bz)|\; \geq\; t\right)\s
\\
\lte \mbP\left(\max\limits_{\bz\in\mbS} |g(\bX; \btheta, \bz) - \mbE\, g(\bX; \btheta, \bz)| + \max\limits_{\bz\in\mbS} |\mbE\, g(\bX; \btheta, \bz)|\; \geq\; t\right)\s
\\
\lte \mbP\left(\max\limits_{\bz\in\mbS} |g(\bX; \btheta, \bz) - \mbE\, g(\bX; \btheta, \bz)|\; \geq\; \dfrac{t}{2}\right) + \mbP\left(\max\limits_{\bz\in\mbS} |\mbE\, g(\bX; \btheta, \bz)|\; \geq\; \dfrac{t}{2}\right).
\ee
Choose $\epsilon \in (0, 1)$ as small as desired and let $c > 0$.
\hide{

\newpage

Remark: Choose
\beno
t
\= 2\, c\; \sqrt{-\log\dfrac{\epsilon}{2} + n \log K}\; \sqrt{K}\; m(\mbS)^2\; \norm{\mA}_\infty^2\; \log n
\ee
or
\beno
\dfrac{t}{2}
\= c\; \sqrt{-\log\dfrac{\epsilon}{2} + n \log K}\; \sqrt{K}\; m(\mbS)^2\; \norm{\mA}_\infty^2\; \log n.
\ee

\newpage

}
Then,
by using
\beno
t
\= 2\, c\; \sqrt{-\log\dfrac{\epsilon}{2} + n \log K}\; \sqrt{K}\; m(\mbS)^2\; \norm{\mA}_\infty^2\; \log n 
&>& 0
\ee
along with Propositions \ref{p.concentration} and \ref{p.expectation},
we can conclude that
\beno
\mbP\left(\max\limits_{\bz\in\mbS} |g(\bX; \btheta, \bz)|\; \geq\; 2\, c\; \sqrt{-\log\dfrac{\epsilon}{2} + n \log K}\; \sqrt{K}\; m(\mbS)^2\; \norm{\mA}_\infty^2\; \log n\right)
\lte \epsilon.
\ee

\end{document}